\long\def\drop#1{}
\documentclass[12pt,a4paper]{report}
\usepackage{hyperref}
\usepackage{a4wide}
\usepackage{amsfonts,amssymb,amsthm,amsmath,amsbsy,manfnt,bbm}
\usepackage{graphicx,pinlabel}
\usepackage{framed}
\usepackage{mathtools}
\usepackage{mathrsfs}
\usepackage{enumitem}
\usepackage{tikz-cd}
\usetikzlibrary{arrows}
\usepackage{float}
\mathchardef\ordinarycolon\mathcode`\:
\mathcode`\:=\string"8000
\begingroup \catcode`\:=\active
  \gdef:{\mathrel{\mathop\ordinarycolon}}
\endgroup
\usepackage{makeidx}\makeindex
\def\iindex#1{\index{#1}#1}

\usepackage{color}
\long\def\red#1{}

\def\R{\mathbb R}

\def\X{\mathcal X}
\def\Y{\mathcal Y}
\def\aD{\mathcal D}
\def\aZ{\mathcal Z}
\def\aG{\mathcal G}
\def\aK{\mathcal K}

\def\E{\mathcal E}
\def\aF{\mathcal F}
\def\P{\mathscr P}
\def\aP{\mathcal P}
\def\M{\mathcal M}
\def\B{\mathcal B}
\def\cD{\mathcal D}
\def\H{\mathcal H}
\def\D{\mathrm{D}}

\def\Vb{V_b} 
\def\Vi{V_i} 

\def\Tan#1{\mathop{\mathrm{Tan}_{#1}}}

\def\isdef{\coloneqq}

\def\bc{\boldsymbol c}
\def\bj{\boldsymbol j}
\def\bw{\boldsymbol w}
\DeclareMathOperator\Ent{\mathsf{Ent}}

\DeclareMathOperator\Prob{Prob}
\DeclareMathOperator\diag{diag}
\DeclareMathOperator*\argmin{argmin}
\DeclareMathOperator\range{range}
\def\Lebesgue{\mathcal L}
\providecommand{\abs}[1]{\left|#1\right|}
\let\e\varepsilon
\def\d{\mathrm d}
\def\un#1{\ensuremath{\mathop{\mathrm{#1\strut}}}}
\def\dual#1#2{\langle #1,#2\rangle}

\def\div{\mathop{\mathrm{div}}}

\DeclareMathOperator{\supp}{supp}

\let\ds\displaystyle
\def\nb#1{\par\medskip\noindent\textbf{#1}}

\newtheorem{theorem}{Theorem}
\newtheorem{lemma}[theorem]{Lemma}
\newtheorem{definition}[theorem]{Definition}
\newenvironment{remark}%
  {\par\medbreak\refstepcounter{equation}%
    \noindent\textbf{Remark~\thelemma. }}%
  {\qed\par\medskip}

\def\XXint#1#2#3{{\setbox0=\hbox{$#1{#2#3}{\int}$}
     \vcenter{\hbox{$#2#3$}}\kern-.5\wd0}}

\newenvironment{danger}{%
  \medbreak\par\clubpenalty=10000
  \footnotesize}{\par\medbreak}
    
\definecolor{shadecolor}{rgb}{0.95,0.95,0.95}
\newenvironment{example}%
{
\begin{MakeFramed}{\FrameSep1cm\advance\hsize-\width \FrameRestore}}%
{\end{MakeFramed}}

\makeatletter
\newcounter{exercise}
\@addtoreset{exercise}{section}
\def\theexercise{\arabic{chapter}.\arabic{section}.\arabic{exercise}}
\newenvironment{exercise}{
  \refstepcounter{exercise}
  \par\medskip\noindent
  \textbf{Exercise \theexercise. }}%
  {\par\medbreak}
\makeatother

\definecolor{MyGray}{rgb}{0.92,0.99,0.98}
\makeatletter
\newdimen\grayboxwidth
\newdimen\grayboxmargin
   {\end{minipage}\end{lrbox}%
   \fboxsep0pt\fbox{\fboxsep\grayboxmargin\colorbox{MyGray}{\usebox{\@tempboxa}}}
}

\makeatother

\begin{document}

\title{Variational Modelling:\\ Energies, gradient flows, and large deviations}
\author{Mark Peletier}

\date{\today\\Version 1.2
\\[1cm]
}
\pagenumbering{Alph}
\maketitle

\pagenumbering{arabic}

\section*{Preface}


\section*{Acknowledgements}
Thanks to David Bourne, Alex Cox, Manh Hong Duong,  Joep Evers, Alexander Mielke, Georg Prokert,  Michiel Renger, and Erlend Storr\o sten for comments on earlier drafts of these notes; and thanks to Upanshu Sharma for contributing Section~\ref{sec:degeneracy-dynamics}.

\section*{Notation}

\begin{minipage}{12cm}
\begin{small}
\begin{tabular}{lll}
$\|\cdot\|_{1,\rho}$ and $(\cdot,\cdot)_{1,\rho}$ & dual Wasserstein norm  and inner product & Sec.~\ref{sec:WassersteinMetricTensor}\\
$\|\cdot\|_{-1,\rho}$ and $(\cdot,\cdot)_{-1,\rho}$ & Wasserstein norm and inner product & Sec.~\ref{sec:WassersteinMetricTensor}\\

$C^k_b(\Omega)$ & space of functions with $k$ ct. and bounded derivatives\\
$\Ent$          & entropy $\rho \mapsto \int \rho\log\rho$ & \eqref{def:Ent}\\
$\H(\rho|\mu)$  & relative entropy of $\rho$ with respect to $\mu$ & Def.~\ref{def:H}\\
$k$				& Boltzmann's constant; $k=1.3806488\cdot 10^{-23} \un{m^2\, kg\, s^{-2}\, K^{-1}}$\\
$\Lebesgue^d$   & $d$-dimensional Lebesgue measure\\
$\M(\Omega)$    & Borel regular, locally finite measures on $\Omega$ & Appendix~\ref{app:measure-theory}\\
$N_A$			& Avogadro's number; $N_A = 6.0221413\cdot 10^{23}$\\
$\P(\Omega)$    & probability measures on $\Omega$ & Appendix~\ref{app:measure-theory}\\
$P_z\aZ$ and $P\aZ$ & process space at $z\in \aZ$ and process bundle & Sec.~\ref{sec:tangents-and-processes}\\
$\aP$           & mapping from process space to tangent space & Sec.~\ref{sec:tangents-and-processes}\\
$R$             & universal gas constant, $R =kN_A \simeq 8.314 \,\un{J\, K^{-1}\,mol^{-1}}$\\ 
$T_z\aZ$, $T\aZ$ & tangent space at $z\in \aZ$ and tangent bundle & Sec.~\ref{sec:tangents}\\        
$W_2(\rho_0,\rho_1)$ & Wasserstein distance between $\rho_0$ and $\rho_1$ & Sec.~\ref{sec:WassersteinDistance}\\
$\Psi$, $\Psi^*$ & dissipation and dual dissipation potentials & Secs.~\ref{sec:GF-intro}, \ref{sec:GF-formal}\\
$\aZ$           & abstract name of the state space 
\end{tabular}
\end{small}
\end{minipage}

\bigskip


\red{Need to think about the confusion surrounding the concept of energy. Two solutions: 1. Include a separate section detailing all the different energetic concepts? 2. Consistently use 'driving functional' for the thing driving a gradient flow?}

\red{Meer grijze blokjes!}

\red{Need to give $R$ its own font to distinguish it from the radius $R$}

\tableofcontents

\chapter{Introduction}
\section{Variational modelling}

Modelling is the art of taking a real-world situation and constructing some mathematics to describe it. It's an art rather than a science, because it involves choices that can not be rationally justified. Personal preferences are important, and `taste' plays a major role. 

Having said that, some choices are more rational than others. And we often also need to explain our choices to ourselves and to our colleagues. In these notes we describe a framework for making and explaining such choices, a framework that we call \emph{variational modelling}, which provides both a workflow for deriving a model and a body of theory to support it. 

We restrict ourselves to the class of \emph{dissipative systems}, which are systems driven by dissipation---dissipation of entropy or free energy, for instance---and more precisely, to the subclass of these systems with the mathematical structure of \emph{gradient flows}. In its simplest form, and in abstract terms, variational modelling of such a  system consists of choosing a state space $\aZ$, a driving functional $\aF$, and a \emph{dissipation mechanism} $\aD$. Once these three are chosen, the evolution equations can be derived, typically in the form of differential or differential-integral equations. The crucial point is that the triple $(\aZ,\aF,\aD)$ contains \emph{all the modelling choices}: it completely determines the model. 

Therefore all discussion about the modelling necessarily focuses on the choice of $(\aZ,\aF,\aD)$:
\begin{quote}
\emph{What are the modelling assumptions that underlie these choices? }
\end{quote}
In these notes we will explain various possible choices, discuss their features, and provide the background and theory that we will need to give a good answer to this question. We will also discuss various generalizations of the simple version above, for instance to `generalized' gradient flows and systems with inertia. Finally, throughout the notes we will give many examples to illustrate the concepts. 

\section{Gradient flows}
\label{sec:GF-intro}
\emph{Gradient flows} are systems that can be written in one of the forms
\begin{subequations}
\label{eq:GF-intro}
\begin{align}
\label{eq:GF-intro-K}
\dot z &= -\aK(z) \aF'(z) \qquad \text{or}\\
\aG(z)\dot z &= - \aF'(z).
\label{eq:GF-intro-G}
\end{align}
\end{subequations}
Here $z\in\aZ$ is the state of the system, $\aF'(z)$ is the (Fr\'echet) derivative of $\aF$ at $z$, and the linear operators $\aK$ and~$\aG$ are two alternative ways of characterizing the dissipation mechanism  $\aD$ mentioned above. Both are assumed to be non-negative and symmetric (all these notions are made concrete in Chapter~\ref{ch:GF}). Whenever $\aK$ and $\aG$ are invertible we assume $\aK = \aG^{-1}$, and the two formulations above are exactly equivalent.

Such systems are driven by the functional $\aF$: the functional $\aF$ decreases along solutions, since the non-negativity of $\aK$ and $\aG$ implies that 
\begin{equation}
\label{ineq:dissipation-intro}
\frac d{dt} \aF(z(t)) = \aF'\cdot \dot z = - \aF' \cdot \aK(z)\aF' = -\dot z \cdot \aG(z) \dot z \leq 0.
\end{equation}
In fact the evolution can be assumed to decrease $\aF$ as fast as possible, where the meaning of \emph{as fast as possible} is given by the operators $\aK$ and $\aG$, i.e. by the dissipation mechanism (see Chapter~\ref{ch:GF}). 

In the modelling of real-world systems, $\aF$ is often an \emph{energy}, which makes equation~\eqref{eq:GF-formal-G} a \emph{force balance:} the right-hand side is the derivative of the energy, i.e.\ a generalized force. The left-hand side can be interpreted as the (generalized) force necessary for the instantaneous (generalized) velocity $\dot z$.

The right-hand side of~\eqref{ineq:dissipation-intro} contains two quadratic forms that will often return. These are related to the \emph{dissipation potential} and the \emph{dual dissipation potential}
\[
\Psi(z,s) := \frac12 s\cdot \aG(z)s \qquad\text{and}\qquad 
 \Psi^*(z,\xi) := \frac12 \xi \cdot \aK(z) \xi.
\]
For fixed $z$, these are dual quadratic forms in $s$ and $\xi$ in the sense of Legendre transforms (see~\eqref{id:duality-GK-PsiPsistar}), 
\begin{align*}
\Psi^*(z,\xi) &= \sup_s \bigl[ \xi\cdot s - \Psi(z,s)\bigr]\\
\Psi(z,s) &= \sup_\xi \bigl[ \xi\cdot s - \Psi^*(z,\xi)\bigr].
\end{align*}
This relation between $\Psi$ and $\Psi^*$ gives a natural generalization of the relationship between $\aG$~and~$\aK$ to the case of non-invertible operators. It also shows, in a very general sense, how it is equivalent to choose any one of the four objects $\aG$, $\aK$, $\Psi$, or $\Psi^*$---the other three follow automatically. 

A central tool in these notes is a reformulation of~\eqref{eq:GF-intro-G} in terms of $\Psi$. 
At any given point $z$,  equation~\eqref{eq:GF-intro-G} is the stationarity condition associated with the minimization problem
\begin{equation}
\label{minpb:GF-general}
\inf_s \Psi(z,s) + \aF'(z)\cdot s.
\end{equation}
If~$\aG$ is invertible, then this expression is strictly convex in $s$ for fixed $z$, and minimizers of~\eqref{minpb:GF-general} are unique. The minimizer $s$ is then equal to the unique solution $\dot z$ of~\eqref{eq:GF-intro-G}. On the other hand, when $\aG$ is not invertible, then~\eqref{minpb:GF-general} suggests a natural extension (see Section~\ref{sec:GF-formal}).


\section{Making choices}

The modelling procedure that we follow in these notes consists of four steps:
\begin{enumerate}
\item Choose \emph{state space} $\aZ$;
\item Choose a {functional} $\aF$ that \emph{drives} the evolution;
\item \label{choice:dissipation} Choose the \emph{dissipation mechanism}, expressed by $\aG$, $\aK$, $\Psi$, or $\Psi^*$;
\item Derive the equations.
\end{enumerate}
Often, Step~\ref{choice:dissipation} consists of two parts:
\begin{enumerate}
\item[\ref{choice:dissipation}a.] Choose the space of all \emph{processes} that participate in changing the state;
\item[\ref{choice:dissipation}b.] Choose the \emph{dissipation potential} $\Psi$ as a function of these processes.
\end{enumerate}

The purpose of these notes is to describe how to do this for a reasonable class of systems, and to explain how one can understand the choices that are made. 

\medskip

One important choice is not yet listed: the choice to use a gradient-flow, or a generalized gradient flow (Section~\ref{sec:GF-formal}) to describe the system. In practice this choice means two things: first, that inertia plays no role, and second, that the coordinate frame is inertial. 

Inertia becomes unimportant when viscous or frictional forces dominate the inertial forces. We work out an example of this in the first system in Chapter~\ref{ch:examples}. 
Gradient-flow-like systems that \emph{do} include inertia have been developed in various forms;  currently the most popular version is called GENERIC (General Equation for Non-Equilibrium Reversible-Irreversible Coupling~\cite{Oettinger05}). It includes generalized gradient flows as a special case, and like generalized gradient flows, provides a form of inherent thermodynamic consistency.

Choosing an inertial frame is important in order to avoid ficticious forces. When the energy is minimal, its derivative is zero, and by~\eqref{eq:GF-intro} therefore the system should be stationary. For instance, an object that is stationary in an inertial frame is rotating when viewed from a rotating frame, and therefore experiences a ficticious, centripetal force---even though the energy might be minimal.

\bigskip

\textbf{Note:} Often the functional $\aF$ that we choose will be the `free energy' that is well known in thermodynamics and statistical mechanics. Many modellers, indeed, would start by postulating that this free energy drives the evolution, without further comment. In these notes we want to presume very little, and therefore we will show \emph{when}, \emph{how}, and \emph{why} the free energy drives the evolution.

\section{Overview}

These notes are structured as follows.
\begin{itemize}
\item Chapter~\ref{ch:examples} describes the application of the variational-modelling framework to three examples: a spring-dashpot system, the slow flow of a viscous fluid, and the diffusion of solutes in water. 
\item In Chapter~\ref{ch:GF} we recall some of the basic facts of gradient flows, including the Wasserstein metric and Wasserstein gradient flows.
\item The third example in Chapter~\ref{ch:examples}, on diffusion, involves the concepts of entropy, free energy, and Wasserstein dissipation. In Chapter~\ref{ch:entropy-free-energy-stationary} we study entropy and free energy, and create an understanding for why they appear in the driving functional of a gradient-flow system.
In Chapter~\ref{ch:Wasserstein-dissipation} we do the same for the Wasserstein dissipation.
\item In Chapter~\ref{ch:further-examples} we discuss further examples of the theory.
\item In Chapter~\ref{ch:thermodynamics} we comment on the relation between the variational-modelling framework on one hand and thermodynamics on the other.
\end{itemize}
We conclude with an appendix on the concepts of measure theory that we use.

\chapter{Examples}
\label{ch:examples}

In this first chapter we discuss three examples---models whose main purpose is to provide illustrations of the general principles in a simple setting. The variational-modelling framework is overpowered for these simple systems, and the modelling may therefore feel slightly clumsy. The benefit of the simplicity is of course that we can concentrate on the essence, and show in a simple, well-known setting, how the framework functions. 

The examples are (a) a spring-dashpot system, (b) slow viscous flow, and (c) diffusion of solutes in a fluid. The third example, on diffusion, will  figure as a central example in the rest of these notes, since it illustrates the important concepts of entropy and  Wasserstein dissipation.

\nb{Dimensions:} Throughout these notes, we work with mathematical objects that have dimensions. As a result all the usual physical constants will appear, and  these constants also help in understanding the origin and interpretation of various terms. Of course, once a model has been constructed, it is often wise to scale the variables and parameters in a different way to facilitate analysis and numerical approximation.

\section{A spring-dashpot system}

As described in the previous chapter, in the variational-modelling framework the construction of a model consists of choosing three entities: the state space, the driving functional, and the dissipation mechanism. In our first example we illustrate this on a very simple system: a spring and a dashpot that are connected---see Figure~\ref{fig:springdashpot}. 

\begin{figure}[htb]
\labellist
\pinlabel $x$ by 0 2 at 54 3
\endlabellist
\centering
\includegraphics[width=5cm]{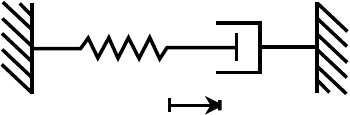}
\caption{A simple spring-dashpot system. The variable $x$ is the extension of the spring and at the same time the position of the dashpot.}
\label{fig:springdashpot}
\end{figure}

\nb{State space: } We characterize the state of the system by the single variable $x\in \R$, the extension of both the spring and the dashpot. 

\nb{Energy: } We choose for the energy of this system the potential energy of the spring $\aF(x)$; for a linear spring this would be $kx^2/2$, where $k$ is the spring constant. 

\nb{Dissipation: } We assume that energy is dissipated in the dashpot, and only there. We assume that the dashpot is linear, which means that the force $f$ on the dashpot and the velocity $v$ of the plunger are related by $f = \eta v$, where $\eta$ is a constant characterizing the dashpot.  Since $\aG\dot z$ should be interpreted as the force necessary to move the system with velocity $\dot z$ (see Section~\ref{sec:GF-intro}), the relation $f=\eta v$ is  the same as $\aG v = \eta v$ for any $v\in\R$. Correspondingly $\aK \xi = \xi/\eta$ for all $\xi$. In this case $\Psi(x,s) = \eta s^2/2$ and $\Psi^*(x,\xi) = \xi^2/2\eta$.

\nb{Equation: }
Applying~\eqref{eq:GF-intro-K} we find the equation $\dot x = -\aK\aF'(x) = -\eta^{-1} \aF'(x)$. For the linear spring $\aF(x) = kx^2/2$ it becomes $\dot x = -\eta^{-1}kx$.

\subsection*{Discussion}

\paragraph{Assumptions and conclusions.}
The list of boldfaced terms above highlights the assumptions that we made. In this context, the concepts and assumptions are fairly standard: the force of the spring is the (negative) derivative of the energy, and the velocity of the plunger in the dashpot is linearly related to the force on the plunger. 

\paragraph{No inertia.}
As we mentioned above, an essential assumption is that inertia plays no role---or more accurately, that inertia can be neglected. In practice this comes down to requiring that the viscous forces dominate the inertial forces in the system. 

As an example, consider the case of a spring-dashpot system with mass, for instance (Figure~\ref{fig:massspringdashpot}).
\begin{figure}[htb]
\labellist
\setlabel{x}{94}60{2}
\setlabel{k}{51}{40}0{-2}
\setlabel{\eta}{48}10{2}
\setlabel{m}{100}{23}00
\endlabellist
\centering
\includegraphics[width=5cm]{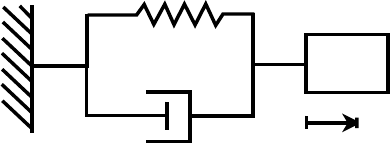}
\caption{A non-gradient-flow system: a mass connected to a spring and a dashpot}
\label{fig:massspringdashpot}
\end{figure}
Assuming a linear spring (with constant $k$) and a linear dashpot (with constant $\eta$), the equation for the position $x$ of the mass is 
\[
m\ddot x + \eta \dot x + k x = 0.
\]
Rescaling this equation by defining a dimensionless time $\tau$ in terms of the old time $t$, $\tau = t \,\eta/k$, we find
\[
\frac{mk}{\eta^2} \partial_{\tau\tau} x + \partial_\tau x +  x = 0.
\]
When the dimensionless combination $mk/\eta^2$ is small, this equation is a singular perturbation of the viscous equation $\partial_\tau x +x=0$, or in dimensional terms, $\eta\dot x + kx=0$. Apart from a fast accommodation at time zero, the solutions of the full system follow the solution of the `outer equation' $\eta\dot x + kx=0$  (see e.g.~\cite[Ch.~5]{Hinch91}).

\section{Slow viscous flow}
\label{sec:slow-viscous-flow}

%

The second example is that of the slow flow of an incompressible, highly viscous liquid, with variable mass density, driven by gravity. The word `slow' means that inertia will be neglected, and then a gradient-flow structure  is natural.

\nb{State space: }We choose a bounded set $\Omega\subset \R^3$ to be the container with the fluid. 
The state variable will be the mass density $\rho:\Omega\to[0,\infty)$ (in $\un{kg \, m^{-3}}$). The state space is therefore $\aZ:= \{\rho\in L^1(\Omega): \rho\geq0\}$. Positions in $\Omega$ are labelled $x$.

\nb{Energy: }We choose as driving functional for this system the gravitational potential
\[
\aF: \aZ \to \R, \qquad \aF(\rho) := g\int_\Omega x_3\rho(x)\, dx,
\]
where $g \approx 9.8\, \un{m\, s^{-2}}$ is the gravitational acceleration and where we choose the last coordinate $x_3$ of $x\in \R^3$ to be the vertical coordinate. 

\nb{Processes: }The system can evolve through a \emph{flow field} $u:\Omega\to\R^3$ satisfying the incompressibility condition $\div u = 0$ in $\Omega$ and the no-slip condition $u=0$ at $\partial \Omega$. For a given flow field $u$, the density $\rho$ evolves according to 
\begin{equation}
\label{eq:proc-tang-conn-Stokes}
\dot \rho + \div \rho u = 0 \qquad \text{in }\Omega.
\end{equation}
Note that since $u=0$ on the boundary $\partial\Omega$, this evolution preserves the total mass $\int_\Omega\rho$. We call the flow field $u$ a \emph{process vector}, since it characterizes all possible processes in this model; it generates a corresponding rate of change in the state, $\dot \rho$, through~\eqref{eq:proc-tang-conn-Stokes}.

\nb{Dissipation mechanism: }Energy is dissipated by viscous friction: the layers of liquid sliding over each other generate heat. This dissipation is characterized by the dissipation potential
\[
\widetilde \Psi(u) := \frac\eta 2 \int_\Omega |\e(u)(x)|^2 \, dx,
\]
where $\eta>0$ is a parameter (the `dynamic viscosity' of the fluid, in $\un{kg\,m^{-1}\, s^{-1}}$), and $\e(u) := \frac12 (\nabla u + \nabla u^T)$ is the symmetric part of $\nabla u$. This function $\widetilde\Psi$ is one-half of the actual dissipation of energy (see e.g.~\cite[(16.3)]{LandauLifshitzVI87} and the discussion below).

\nb{Equations: }The resulting equations are now determined by solving the minimization problem
\begin{equation}
\label{minpb:Stokes-intro}
\inf_u \widetilde \Psi(u) + \langle \aF'(\rho),-\div\rho u\rangle,
\end{equation}
We introduced this minimization problem in Section~\ref{sec:GF-intro}, in the case of a dissipation potential $\Psi$ defined on tangent vectors $\dot \rho$, instead of the potential $\widetilde \Psi$ defined on flow fields $u$. We discuss the difference between these two in Section~\ref{sec:tangents-and-processes}. The dual pairing is the pairing of the Fr\'echet derivative $\aF'$ with a tangent vector (Section~\ref{sec:tangents}).

Writing the expression in~\eqref{minpb:Stokes-intro} as
\[
\frac\eta 2 \int_\Omega |\e(u)|^2 - g \int_\Omega x_3 \div\rho(x) u(x) \, dx
= \frac\eta 2 \int_\Omega |\e(u)|^2 + g \int_\Omega e_3 \rho u,
\]
where $e_3$ is the unit vector for the third coordinate, we calculate the stationarity condition for~\eqref{minpb:Stokes-intro} to be
\[
\forall \tilde u: \qquad 0 = \eta\int_\Omega \e(u) : \e(\tilde u)  + g \int_\Omega e_3 \rho \tilde u - \int_\Omega p \div \tilde u,
\]
where $p$ is a Lagrange multiplier associated with the condition $\div u = 0$, and the identity holds for all $\tilde u$ that vanish on $\partial \Omega$ (even those that are not divergence-free). 
By applying partial integration we deduce that
\begin{alignat*}2
-\div\, [\eta \e(u) - pI] &= -e_3\rho  & \qquad & \text{in }\Omega,\\
\div u &= 0&& \text{in }\Omega,\\
u &= 0 && \text{on }\partial \Omega.
\end{alignat*}

\section*{Discussion}

\paragraph{Modelling choices.}
As in all the examples of these notes, the lack of inertial effects is important. In this case it brings us to the Stokes equations instead of the Navier-Stokes equations.

Note how in this case we choose the state space $\aZ$ separately from the space of \emph{processes}. This is a feature that we will encounter more often, since from a modelling point of view the \emph{state} often has little in common with the \emph{mechanisms} that can change  the state.

\paragraph{Quantifying the dissipation mechanism.}
In the spring-dashpot example we could specify either $\aG$ or $\aK$; both were relatively simple to characterize. In this case it is more straightforward to specify $\Psi$ than any of the other three possibilities $\aG$, $\aK$, and $\Psi^*$.

In the case of quadratic dissipation potentials, as in this example, there is a useful connection between the dissipation \emph{potential} and the actual dissipation $D$ of energy. This deserves its own discussion.

\paragraph{Dissipation and dissipation potential}
\label{dissipation-dissipation-potential}
A consequence of the gradient-flow, or dissipative structure that we use here, is the free-energy-dissipation identity
\[
\aF(\rho_T) +\int_0^T D(u_t)\, dt  = \aF(\rho_0),
\]
or the differential version,
\begin{equation}
\label{eq:energy-dissipation-differential}
\partial_t \aF(\rho_t) = -D(u_t),
\end{equation}
where we use subscripts to indicate dependence on time. Here $u_t$ is assumed to be the minimizer of~\eqref{minpb:Stokes-intro}, and the \emph{dissipation} $D$ is defined as 
\[
D(u) := \eta \int_\Omega {|\e(u)|^2}\, dx.
\]
By~\eqref{eq:energy-dissipation-differential} $D$ is indeed the dissipation of energy $\aF$ per unit of time. 

The \emph{dissipation} $D$ and the \emph{dissipation potential} $\Psi$ are closely related, but different. The dissipation $D$ quantifies the actual conversion of  energy into heat; the dissipation potential $\Psi$ is the object in the minimization problem~\eqref{minpb:Stokes-intro}. The general formula is $D(u) = \langle \partial_u\Psi(u),u\rangle$, and when $D$ and $\Psi$ are quadratic in $u$, this expression simplifies to $D = 2\Psi$.

\section{Diffusion of particles in a fluid}
\label{sec:diffusion-particles-fluid}

Our third example will be central in the discussion about `how' and `why': it will be the main example on which we illustrate the methods. As in  the previous examples, we already know the equations for this system, but we will re-derive them using the variational framework. We will make choices for the state space, the driving functional, and the dissipation structure, and we will show how they reproduce the known equations for this system. 

For this system, however, the choice of driving functional and dissipation are not immediately clear---and we will not explain them in this chapter. Indeed, this is exactly the \emph{raison d'\^etre} of these notes: the variational modelling philosophy is powerful, but the assumptions underlying it are not always obvious. 
The purpose of this example is therefore to provide a test system, on which we will develop and test the understanding in Chapters~\ref{ch:entropy-free-energy-stationary} and~\ref{ch:Wasserstein-dissipation}.

\bigskip

We model the diffusion of a collection of particles in a stationary viscous fluid under the influence of gravity and thermal agitation.

\nb{State space: } We choose a bounded set $\Omega\subset \R^d$ to be the container with the fluid; the particles are represented by their concentration $c:\Omega \to [0,\infty)$ (in moles per $\un {m^3}$). The state space is therefore $\aZ:= \{c\in L^1(\Omega): c\geq0\}$. Positions in $\Omega$ are labeled $x$.

This choice represents many assumptions at the same time: the particles are indistinguishable from each other, there are many of them, we are only interested in their density, and the fluid remains stationary; and the fact that we did not include temperature as a variable implies that temperature is held constant. 

\nb{Energy: } We choose as driving functional for this system the following combination of \emph{entropy} and \emph{gravitational energy}:
\begin{equation}
\label{choice:entropy-diffusion-particles-fluid}
\aF: \aZ\to\R, \qquad \aF (c) := RT \int_\Omega c(x)\log \frac{c(x)}{c_0}\, dx
+ \rho g\int_\Omega x_3 c(x) \, dx.
\end{equation}
Here $c_0>0$ is an arbitrary reference concentration (since it is arbitrary, it should not appear in the final equations; and indeed it does not), $R \simeq 8.314 \,\un{J}\un K^{-1}\un{mol}^{-1}$ is the \emph{universal gas constant} and $T$ is the temperature; $\rho$ is the mass density contrast of the particles (the difference between the density of the particles and that of the fluid, in $\un{kg/m^3}$), and $g \simeq 9.8 \un {m/s^2}$ is the gravitational acceleration.

This choice follows from assuming that the particles have no spatial preference, that they are thermally agitated at temperature $T$, that there are many of them, and that they are independent in a certain way. How these aspects can be recognized in this choice is indeed a  non-trivial question, which is treated in detail in Chapters~\ref{ch:entropy-free-energy-stationary} and~\ref{ch:Wasserstein-dissipation}. For now we continue with the next part of the modelling.

\medskip


\nb{Processes: }Similarly to~\eqref{eq:proc-tang-conn-Stokes} we allow the state $c\in \aZ$ to change through the effects of a flux $w:\Omega \to\R^d$  by
\begin{equation}
\label{modass:c-w}
\dot  c + \div cw =0  \quad \text{in $\Omega$},\qquad w\cdot n = 0 \quad \text{on }\partial\Omega,
\end{equation}
or in weak form,
\[
\forall \varphi \in C^1_b(\Omega): \qquad
\partial_t \int_{\Omega} \varphi(x) c(t,x)\, dx - \int_\Omega c(t,x)w(t,x)\nabla \varphi(x)\, dx =0.
\]

This is again a modelling choice: for instance, the divergence form of~\eqref{modass:c-w} implies that no particles can be created or removed. If one wants to allow for chemical reactions, for instance, which create and destroy partices, then the characterization~\eqref{modass:c-w} should also contain terms that are not in divergence form. The boundary condition $w\cdot n = 0$ is the modelling assumption that no particles enter or leave through the boundary.

The set of all \emph{process vectors} is therefore the set $\{w:\Omega\to\R^d: w\cdot n = 0 \text{ on } \partial \Omega\}$. Each process vector $w$ generates a \emph{tangent vector} $\dot c$ through the relation~\eqref{modass:c-w}.


\medskip
Now that we have characterized the processes that change the state, we can choose the dissipation:
\nb{Dissipation potential: }We define the \emph{dissipation potential} on the set of process vectors~$w$ as the functional 
\begin{equation}
\label{def:dissipation-potential-diffusion}
\widetilde\Psi(c;w) := \frac\eta2 \int_\Omega c(x){|w(x)|^2}\, dx.
\end{equation}
This choice follows from assuming that the fluid remains stationary and is linearly viscous, that the particles are sufficiently symmetric, and that they are far enough from each other to prevent interference between the flow fields surrounding them. Here the constant $\eta>0$ is a \emph{friction} coefficient: it characterizes the ratio between a velocity of a particle through the fluid and the corresponding force exerted by the fluid on the particle. 
As in the case of the entropy, the origin of this choice, and the interpretation of $\eta$,  will be discussed at length in Chapter~\ref{ch:Wasserstein-dissipation}.

\begin{danger} 
As in the previous example, the dissipation potential is one half of the  \emph{actual dissipation} of free energy; see the discussion on page~\pageref{dissipation-dissipation-potential}.
\end{danger}

\medskip

\nb{Equations: }
To derive the equations, we postulate that at a given time $t$ and at state $c$, the process vector $w$ is characterized by the minimization problem
\begin{equation}
\label{min:ex-diffusion}
\min_{w,\dot c}\  \Bigl\{ \widetilde \Psi(c;w) + \langle \aF'(c), \dot c\rangle:\ \dot c \text{ and } w \text{ connected by }\eqref{modass:c-w}\Bigr\},
\end{equation}
as we also did in the previous example.

We find the equation for $c$ from the stationarity condition: writing the expression in braces as
\[
\widetilde \Psi(c;w) + \langle \aF'(c),-\div w\rangle = \widetilde \Psi(c;w) + RT \int_\Omega \nabla \Bigl(\log \frac c{c_0}+1 \Bigr) \cdot w,
\]
the stationarity condition is seen to be 
\[
\forall \tilde w: \qquad 0 = \int_\Omega c {w\tilde w} \, dx + RT\int_\Omega \tilde w{\nabla c} \, dx, 
\]
which implies $cw = -RT/\eta \, \nabla c$ and leads to the equation
\begin{subequations}
\label{eq:diffusion}
\begin{alignat}2
\dot  c &= \frac{RT}\eta \Delta c &\qquad & \text{in }\Omega,\\
\partial_n c &= 0 && \text{on }\partial\Omega.
\end{alignat}
\end{subequations}

\section*{Discussion}

\paragraph{Comparison to classical derivations.} A common way of deriving the diffusion equation~\eqref{eq:diffusion} is by combining the conservation law~\eqref{modass:c-w} with a \emph{constitutive law} $cw = -D\nabla c$, known under the name of \emph{Fick's law}. This leads to the same equation, but with $RT/\eta$ replaced by $D$. 

In this form Fick's law is beautifully simple, and indeed there is little reason to use the overpowered variational-modelling route to model this simple system. However, for more complex systems it is often not clear how to generalize Fick's law, for instance when species interfere with each other's mobility, as in the case of cross-diffusion, or when species diffuse on lower-dimensional structures such as surfaces or linear structures. In addition, for some systems the species concentration has an energetic impact that is not easy to quantify in a simple constitutive law, such as the case when in vesicle membranes the species changes the rigidity of those membranes. 

\red{\paragraph{What happened to Fick's law?} Considering that the variational-modelling route of the previous section leads to the same equation as Fick's law, it is natural to ask how Fick's law is hiding in the variational-modelling choices. We comment on this in Section~\ref{sec:comments-on-Wasserstein-dissipation}.
}

\paragraph{The motivation is not yet complete.} As indicated, we have not yet explained which modelling choices are represented by the chosen free energy $\aF$ and the dissipation potential~$\Psi$. After we provide this motivation in Chapters~\ref{ch:entropy-free-energy-stationary} and~\ref{ch:Wasserstein-dissipation}, this derivation, and those that follow, will be a self-contained modelling route. 

\paragraph{The derivation is formal.} The description of all allowable changes $\dot c$ is formal; we did not specify conditions on $\dot c$ and $w$ in~\eqref{modass:c-w}. In many other situations the situation is even worse, since $\aZ$ is often a set without a differentiable structure, and therefore a `smooth curve' might not even exist. It is to handle this type of difficulty that the theory of metric-space gradient flows has been developed~\cite{AmbrosioGigliSavare08}.

In fact it is very common that modelling derivations are not rigorous, and often would even be very hard to make rigorous. The difficulty usually lies in the assumed regularity of the components of the model: if the solutions are regular, then the arguments are rigorous. Of course there are several well-known examples in which this approach fails to produce complete models; for instance, in the case of weak solutions of scalar conservation laws, uniqueness of such solutions only holds under additional entropy conditions. One should think of these model derivations as thought experiments, whose outcome should later still be checked for mathematical consistency.

\paragraph{The process space has redundant information.}
The field $w$ contains more information than $\dot c$, since adding a divergence-free vector field to $cw$ does not change $\dot c$. This is not a problem; the minimization~\eqref{min:ex-diffusion} deals with this redundancy in a natural way.  In the context of Wasserstein gradient flows this phenomenon is well known, where it leads to the characterization of the Wasserstein tangent in terms of gradients (see e.g.~\cite{Otto01} or \cite[Th.~9.3.2]{AmbrosioGigliSavare08}).

Such redundancies are very common, and they arise from the fact that it is often convenient to characterize the state space and the space of processes in different ways. We will see more examples of this below. 

\bigskip

\red{Maybe add comparison of diffusion coefficients with the numbers generated here?}

\chapter{Gradient-flow facts}
\label{ch:GF}

\section{Forces and velocities, tangents and cotangents}
\label{sec:tangents}

Variational modelling is made clearer by adopting the geometric terminology of tangents and cotangents, even though we do this in a non-rigorous manner. 

\medskip

Let $\aZ$ be a subset of a linear space. The space of tangent vectors at a point $z\in \aZ$, noted $T_{z}\aZ$, is the set of all derivatives $\dot {\boldsymbol z}(0)$ of smooth curves $\boldsymbol z:(-1,1)\to\aZ$ such that $\boldsymbol z(0)=z$. The space of cotangent vectors, noted $T^*_{z}\aZ$, is the set of all linear functionals on $T_{z}\aZ$. The set of all pairs $(z,\tau)$ where $z\in \aZ$ and $\tau\in T_z\aZ$ is called the \emph{tangent bundle} and is denoted $T\aZ$, and similarly the set of all points and cotangent vectors is written $T^*\aZ$. If $\tau\in T_{z}\aZ$ and $\xi\in T^*_z\aZ$, then we write $\dual \xi \tau=\dual\tau\xi\in \R$ for the value of $\xi$ applied to $\tau$. 

In mechanical terminology, the space of tangents is the space of \emph{generalized velocities} or \emph{infinitesimal displacements}. The cotangent space can be identified with \emph{generalized forces}, if one gives $T\aZ$ and $T^*\aZ$ dimensions such that $\dual \xi \tau$ has dimensions of energy per unit time (power); then the duality is sometimes called a \emph{work couple}. The classical example is when $\aZ=\R$ consists of distances, measured in $\un m$, and $T^*\aZ$ has dimensions of force $[\un N]$; then $\tau$ is a velocity $[\un {m/s}]$ and $\dual \xi\tau$ has dimensions $\un{Nm/s}$. But many other combinations are possible, such as when $z$ is an angle, $\tau$ is a rotational velocity $[1/\un s]$, and $\xi$ is a torque $[\un{Nm}]$.

\begin{danger}
This division into tangents and cotangents is useful from many points of view. Maybe the most important is that correct use of this structure guarantees coordinate invariance, as is the case in Riemannian geometry.
\end{danger}

\medskip

If $E:\aZ\to\R$ has dimensions of energy, then the derivative of $E$ with respect to $z$ is a cotangent vector with the dimensions of a generalized force; for any $\tau\in T_z\aZ$, $\dual{E'(z)}\tau $ has dimensions of power, and if $z:(-1,1)\to\aZ$, then $\partial_t E(z(t))|_{t=0} = \dual{E'(z(0))}{\dot z(0)}$. This derivative is mathematically the Fr\'echet derivative. 

\medskip

The operators $\aK$ and $\aG$ that we mentioned in Section~\ref{sec:GF-intro} are examples of \emph{duality maps}, that map tangents to  cotangents or vice versa: for each $z\in\aZ$, $\aG(z)$ is an operator that maps $T_z\aZ$ to $T_z^*\aZ$, and similarly $\aK(z)$ maps $T_z^*\aZ$ to $T_z\aZ$. (Heed the notation: for $z\in\aZ$, $\aG(z)$ is an operator that takes a tangent $\tau\in T_z\aZ$ and maps it to a cotangent vector written as $\aG(z)\tau$.) Because of this property, they give rise to bilinear forms on $T_z\aZ$ and $T^*_z\aZ$, respectively:  
\begin{align*}
(\tau_1,\tau_2)_{\aG,z} &:= \dual {\tau_1}{\aG(z)\tau_2}\\
(\xi_1,\xi_2)_{\aK,z} &:= \dual {\xi_1}{\aK(z)\xi_2}.
\end{align*}
These bilinear forms are symmetric if $\aG$ and $\aK$ are {symmetric}, in the sense that $\dual {\tau_1}{\aG(z)\tau_2}= \dual {\tau_2}{\aG(z)\tau_1}$.

Such duality maps are necessary in equations such as~\eqref{eq:GF-intro}, when one wishes to equate~$\dot z$ with the derivative of an energy, since these two objects live in different spaces. The duality maps implement the mapping from one to the other.

\section{Gradient flows}
\label{sec:GF-formal}

We now revisit the informal definition of a gradient flow that we gave in Section~\ref{sec:GF-intro}. The starting point is to think of gradient flows as a set of equations of the form 
\begin{subequations}
\label{eq:GF-formal}
\begin{align}
\label{eq:GF-formal-K}
\dot z &= -\aK(z) \aF'(z) \qquad \text{or}\\
\aG(z)\dot z &= - \aF'(z).
\label{eq:GF-formal-G}
\end{align}
\end{subequations}
As we described earlier, the interpretation of $\aG$ and $\aK$ is as a characterization of dissipation, through the energy-dissipation identity
\begin{equation}
\label{id:energy-dissipation}
\frac d{dt} \aF(z(t)) = \dual{\aF'}{\dot z} = - \dual{\aF'}{\aK(z)\aF'} = -\dual{\dot z}{\aG(z) \dot z}.
\end{equation}
However, we  need to allow for several generalizations:
\begin{enumerate}
\item \label{need-for-generalization-degenerate} The duality operators $\aK$ and $\aG$ can be degenerate or infinite; for instance, certain changes of state, i.e.\ certain tangent vectors, might not lead to dissipation of energy, implying that $\aG$ vanishes in those directions; similarly, certain directions might be inadmissible, implying that `$\aG=\infty$' in those directions. By duality, these two possibilities translate into $\aK$ being `infinity' in the first case and `zero' in the second, in the corresponding dual directions.
\item \label{need-for-generalization-nonlinear} The duality operators might not be linear in their arguments; a typical example of this is a rate-independent system, in which case $\aG(z)(\lambda s)  = \aG(z) s$ for all $\lambda>0$. 
\item \label{need-for-generalization-process} As we saw in Sections~\ref{sec:slow-viscous-flow} and~\ref{sec:diffusion-particles-fluid}, it is natural to model changes of state through auxiliary variables such as flow fields. This suggests to define dissipation in terms of a \emph{process space}, which can be different than the tangent space.
\end{enumerate}

In the next section we discuss how to deal with point~\ref{need-for-generalization-process}. To deal with~\ref{need-for-generalization-degenerate} and~\ref{need-for-generalization-nonlinear}, we redefine the concept of a gradient flow in terms of a pair of dissipation functionals $\Psi:T\aZ\to\R$ and $\Psi^*:T^*\aZ\to\R$, which are assumed to be dual in the sense that for all $z\in \aZ$, 
\begin{align*}
\Psi^*(z,\xi) &= \sup_{s\in T_z\aZ} \bigl[ \dual{\xi} s - \Psi(z,s)\bigr]\\
\Psi(z,s) &= \sup_{\xi\in T^*_z\aZ} \bigl[ \dual{\xi}s - \Psi^*(z,\xi)\bigr].
\end{align*}
Because of this duality, both $\Psi$ and $\Psi^*$ are convex in their second argument, and we have the inequality
\[
\Psi(z,s) +\Psi^*(z,\xi) \geq \dual \xi s\qquad \text{for all }s\in T_z\aZ\text { and }\xi\in T_z^*\aZ,
\]
with equality if and only if $s\in \partial_\xi \Psi^*(z,\xi) \Longleftrightarrow \xi\in \partial_s \Psi(z,s)$.

Given (linear) duality operators $\aG$ and $\aK$, we define corresponding (quadratic) potentials $\Psi$ and $\Psi^*$ by $\Psi(z,s) := \frac12 \dual{s}{\aG(z)s}$ and $\Psi^*(z,\xi) := \frac12 \dual{\xi}{\aK(z)\xi}$. The fact that these are a dual pair can be seen as follows:
\begin{align}
\Psi^*(z,\xi) &= \sup_{s\in T_z\aZ} \bigl[ \dual\xi s - \Psi(z,s)\bigr]\notag\\
&= \sup_{s\in T_z\aZ} \bigl[ \dual \xi s - \tfrac12 \dual{s}{\aG(z)s}\bigr]\notag\\
&= \sup_{\sigma\in T^*_z\aZ} \bigl[ \dual \xi{\aK(z)\sigma} - \tfrac12 \dual{\aK(z)\sigma} \sigma \bigr] \qquad \text{by substituting }s=\aK(z) \sigma\notag\\
&= \tfrac12 \dual{\aK(z)\xi}\xi +  \sup_{\sigma\in T^*_z\aZ} -\tfrac12 \dual{\aK(z) (\xi-\sigma)}{\xi-\sigma}
= \tfrac12 \dual{\aK(z)\xi} \xi.
\label{id:duality-GK-PsiPsistar}
\end{align}
For this case the identity conditions $s\in \partial_\xi\Psi^*(x,\xi) \Longleftrightarrow \xi\in \partial_s \Psi(z,s)$ simplify to $s = \aK(z)\xi \Longleftrightarrow \xi=\aG(z)s$.

\bigskip

We now give a defintion of a generalized gradient flow that also applies  to the case of non-quadratic potentials $\Psi$ and $\Psi^*$.

\begin{definition}[Global formulation]
\label{def:GF-global}
A function $z:[0,T]\to\aZ$ is a solution of the \emph{generalized gradient flow} associated with $(\aZ,\aF,\Psi)$ if it satisfies the inequality
\begin{equation}
\label{ineq:def-GGF-integral}
\aF(z(T)) - \aF(z(0)) + \int_0^T \Bigl[\Psi\bigl(z(t),\dot z(t)\bigr) + \Psi^*\bigl(z(t),-\aF'(z(t))\bigr)\Bigr]\, dt \leq 0.
\end{equation}
\end{definition}

This definition extends the previous ones:
\begin{lemma}
Assume that $\aZ$ is a Hilbert space. 
Let $\aF:\aZ\to\R$ be smooth, and assume that for each $z\in\aZ$ the duality operators $\aG(z)$ and $\aK(z)$ are invertible with $\aG(z)= \aK(z)^{-1}$. Define $\Psi(z,s) := \frac12 \dual{s}{\aG(z)s}$ and $\Psi^*(z,\xi) := \frac12 \dual{\xi}{\aK(z)\xi}$. Let $z:[0,T]\to\aZ$ be smooth. 
Then the following four properties are equivalent:
\begin{enumerate}
\item $\dot z = -\aK(z)\aF'(z)$ for all $t\in[0,T]$ (i.e., $z$ satisfies~\eqref{eq:GF-formal-K});
\item $\aG(z)\dot z = -\aF'(z)$ for all $t\in[0,T]$ (i.e., $z$ satisfies~\eqref{eq:GF-formal-G});
\item $z$ satisfies~\eqref{ineq:def-GGF-integral};
\item \label{lemma:equivalent-GF-forms:partRayleigh} At each time $t$, and for given $z(t)$, $\dot z(t)$ is given by the property that it minimizes $\Psi(z(t),\dot z) + \dual{\aF'(z(t))}{\dot z}$.
\end{enumerate}
\end{lemma}

\begin{danger}
Note that the conditions of this lemma exclude all but the most trivial of examples, since in nearly all cases $\aF$ will not be smooth with respect to the topology generated by $\aG$; this is exactly the origin of the parabolic behaviour that infinite-dimensional gradient flows typically show. Therefore either $\aF$ will not be smooth in the chosen topology, or $\aG$ will not be invertible with respect to that topology.

The reason for this lemma therefore is not to apply it directly, but to provide a case in which the usually formal connections can be made rigorous.
\end{danger}

\begin{proof}
The equivalence $\eqref{eq:GF-formal-K}\Longleftrightarrow \eqref{eq:GF-formal-G}$ is immediate. To show $\eqref{eq:GF-formal}\Longleftrightarrow\eqref{ineq:def-GGF-integral}$, note that for \emph{any} smooth curve $z:[0,T]\to\aZ$ we have 
\[
\frac d{dt} \aF(z(t)) = \dual{\aF'(z(t))}{\dot z(t)}
\stackrel{(*)}\geq -\Psi\bigl(z(t),\dot z(t)\bigr) - \Psi^*\bigl(z(t),-\aF'(z(t))\bigr).
\]
Therefore inequality~\eqref{ineq:def-GGF-integral} holds with the \emph{opposite} inequality for any curve $z$; imposing the inequality~\eqref{ineq:def-GGF-integral} is equivalent to imposing equality, which in turn is equivalent to imposing equality in $(*)$ above; this in turn is equivalent to $\dot z = \aK(z)(-\aF'(z))$. This proves $\eqref{eq:GF-formal}\Longleftrightarrow\eqref{ineq:def-GGF-integral}$.

Finally, the equation $\aG(z)\dot z = -\aF'(z)$ is the stationarity condition for the minimization problem $\min_s \Psi(z,s) + \dual{\aF'(z)}{s}$, and therefore part~\ref{lemma:equivalent-GF-forms:partRayleigh} is equivalent with~\eqref{eq:GF-formal-G}.
\end{proof}

Part~\ref{lemma:equivalent-GF-forms:partRayleigh} is so important that we emphasize it and give it a name:
\begin{example}
\begin{definition}[Local formulation]
\label{def:GF-local}
A function $z:[0,T]\to\aZ$ is a \emph{local} solution at time~$t$ of the {generalized gradient flow} associated with $(\aZ,\aF,\Psi)$ if $\dot z(t)$ has the property
\begin{equation}
\label{minpb:def-dotz}
\dot z(t) \text{ solves }\quad 
\min_s \Psi(z(t),s) + \dual{\aF'(z(t))}{s}.
\end{equation}
\end{definition}
\end{example}

\section{Tangents and processes}
\label{sec:tangents-and-processes}

We now make the second generalization of the gradient-flow concept, in order to deal with point~\ref{need-for-generalization-process} on page~\pageref{need-for-generalization-process}.

A simplistic reformulation of the variational modelling method is `(a) choose the state space and the energy on this space, and (b) choose how the state can change, and how much energy that costs'. However, the requirements on the state space are different than the requirements on the changes of state: the state space needs to be rich enough to define the energy $\aF$ on this space, and the description of the potential changes of state needs to  be rich enough to define the dissipation.

In practice, from a modelling point of view, the state evolves as the result of physical \emph{processes}, which need not be related to the states themselves in a one-to-one manner. It therefore makes sense to define a separate \emph{space of processes} $P_z\aZ$. For instance,
\begin{itemize}
\item The state space could be that of the volume $z=V>0 $ of fluid in a container, and this volume might change because of an inflow $f_i>0$ and an outflow $f_o>0$. Then the process space is
\[
P_V\aZ = \{(f_o,f_i)\in\R^2\},
\]
and the tangent $\dot V$ is given in terms of $(f_o,f_i)$ by $\dot V = f_i-f_o$.
\item In the example of diffusion of solutes (Section~\ref{sec:diffusion-particles-fluid}), $z=c$, the concentration field, and  the tangent space is the space of derivatives $\dot c$ of smooth curves, which one can formally identify with the subset of $L^1(\Omega)$ with zero integral; the process space is
\[
P_c\aZ = \{ w:\Omega\to\R^3: w\cdot n =0 \text{ on }\partial\Omega\}.
\] 
Given $w\in P_c\aZ$, the corresponding tangent is $\dot c = -\div cw$. 
\end{itemize}

\medskip

The minimization problem~\eqref{minpb:def-dotz} provides a natural way to formulate the dissipation in terms of the process space. Assume that we have chosen a state space $\aZ$, a functional $\aF$, a process space $P_z\aZ$, and a dissipation potential $\Psi(z,w)$ defined on pairs $(z,w)$ of $z\in \aZ$ and $w\in P_z\aZ$. Assume that for each $z\in \aZ$ an operator $\aP(z):P_z\aZ\to T_z\aZ$ is given that maps the process space to the tangent space, such that for each $z\in\aZ$ and $w\in P_z\aZ$, $\aP(z)w$ is the corresponding tangent vector. (In the examples above, for the container volume $\aP(V)(f_o,f_i) = f_i-f_o$, and for the solute diffusion $\aP(c)w = -\div cw$.) We can then define a dissipation potential $\widetilde \Psi$ on the tangent space as 
\begin{equation}
\label{def:Psitilde}
\widetilde \Psi(z,s) := \inf_w  \{\Psi(z,w): s = \aP(z)w\}.
\end{equation}
This dissipation potential $\widetilde \Psi(z,s)$ corresponds to selecting the best process $w$ that gives rise to the tangent vector $s$: that is, `best' in the sense of having the smallest dissipation potential.

Both the global Definition~\ref{def:GF-global} and the local Definition~\ref{def:GF-local} can now be applied to $\widetilde\Psi$. Note that the local definition~\eqref{minpb:def-dotz} can be rewritten as 
\begin{align*}
&\dot z \in \argmin_{s\in T_z\aZ} \widetilde\Psi(z,s) + \dual{\aF'(z)}{s}\\
\Longleftrightarrow \quad& \dot z \in \argmin_{s\in T_z\aZ} \min_{w \in P^{-1}s}\Psi(z,w) + \dual{\aF'(z)}{Pw}\\
\Longleftrightarrow \quad& w \in \argmin_{w \in P_z\aZ}\Psi(z,w) + \dual{\aF'(z)}{Pw}
  \qquad \text{and } \dot z = \aP(z)w.
\end{align*}
We summarize this discussion in the form that we use it throughout these notes:
\begin{example}
If $P_z\aZ$ is the process space at $z\in \aZ$, with dissipation $\Psi(z,w)$, and $\aP:P\aZ\to T\aZ$ maps processes onto tangents, then we can find the evolution equations by solving the minimization problem
\[
\min_{w\in P_z\aZ} \Psi(z,w) + \dual{\aF'(z)}{\aP(z)w}.
\]
Then $\dot z$ is given by the projected minimizer $\aP(z)w$.
\end{example}

\section{Wasserstein distance and Wasserstein gradient flows}

The example of Section~\ref{sec:diffusion-particles-fluid} is also known as a \emph{Wasserstein gradient flow}, and this class of gradient flows is central in these notes. 
A good reference on the Wasserstein metric and Wasserstein gradient flows is~\cite{AmbrosioGigliSavare08}.

\subsection{Wasserstein gradient flows}

Wasserstein gradient flows are gradient flows in the state space $\aZ = \M_2$ of non-negative finite-mass measures with finite second moments, i.e. the space
\[
\M_2(\R^d) := \left\{ \rho\in \M(\R^d): \rho\geq 0 \text{ and } \int_{\R^d} (1+ |x|^2) \,\rho(dx) <\infty\right\}.
\]
See Appendix~\ref{app:measure-theory} for some background on measure-theoretical concepts and notation. 

We define the gradient flow by defining the operator $\aK$. For $\rho\in \M_2(\R^d)$ and a function $\xi:\R^d\to\R$ with weak derivative $\nabla \xi\in L^2(\rho;\R^d)$, let $\aK(\rho)\xi$ be the distribution
\[
\aK(\rho) \xi :=  -\div \rho\nabla \xi .
\]
We interpret this as an operator on the cotangent space of $\aZ$ by applying it not to the Fr\'echet derivative $\aF'$ but to the variational derivative $\D \aF$, defined by
\[
\int_{\R^d} \D\aF(\rho)(x)f(x)\, dx = \dual{\aF'(\rho)}{f} 
  = \lim_{h\to0} \frac1h \bigl[\aF(\rho+hf)-\aF(\rho)\bigr] \qquad\text{for all }f.
\]
The Wasserstein gradient flow of a functional $\aF$ therefore is the evolution equation
\[
\dot \rho = \div \rho\nabla \D\aF(\rho).
\]

\begin{danger}
Many well-known diffusive partial differential equations have the structure of a gradient flow with respect to the Wasserstein metric or a related metric. These are a few examples of energies and corresponding equations:
\begin{alignat*}2
&\aF(\rho) = \int [\rho(x)\log \rho(x) + \rho(x)V(x)]\, dx &\quad& \partial_t \rho = \Delta \rho  + \div \rho\nabla V\\
&\aF(\rho) = \int [U(\rho(x)) + \rho(x)V(x)]\, dx  + \frac12 \iint \rho(x)\rho(y) W(x-y)\, dxdy &\qquad& \partial_t \rho = \div \rho\nabla \bigl[U'(\rho)+ V + W*\rho\bigr]\\
&\aF(\rho) = \int |\nabla \rho|^2(x)\, dx   &\qquad& \partial_t \rho = -\div \rho\nabla \Delta \rho\\
\end{alignat*}
\end{danger}

\subsection{The Wasserstein metric tensor}
\label{sec:WassersteinMetricTensor}

As we discussed in Sections~\ref{sec:GF-intro} and~\ref{sec:GF-formal}, the operator $\aK$ gives rise to an inner product, which is actually a weighted Sobolev inner product,
\[
(\xi_1,\xi_2)_{1,\rho} := \int_{\R^d} \xi_1 \aK(\rho)\xi_2 
= \int_{\R^d} \nabla \xi_1\nabla \xi_2 \,d\rho,
\]
with corresponding norm
\begin{equation}
\label{def:dualWassersteinNorm}
\|\xi\|_{1,\rho}^2  = \int_{\R^d} |\nabla\xi|^2 \, d\rho.
\end{equation}
In Wasserstein terminology this is called the \emph{dual} norm. 
The corresponding \emph{primal} Wasserstein norm $\|\cdot \|_{-1,\rho}$ is defined for any distribution $s\in\cD'(\R^d)$,
\begin{equation}
\label{def:locWM}
\|s\|_{-1,\rho}^2 := \inf_v \left\{\int_{\R^d}|v|^2 \,d\rho:
v\in L^2(\rho;\R^d), \ s +\div \rho v =0 \text{ in } \cD'(\R^d)\right\}.
\end{equation}
If $s$ has no representation as $\div \rho v$ for any $v\in L^2(\rho)$, then the norm is given the value $+\infty$.

Note how definition~\eqref{def:locWM} contains the same expression as the dissipation potential $\widetilde\Psi$ in~\eqref{def:dissipation-potential-diffusion}; in fact~\eqref{def:locWM} is an example of~\eqref{def:Psitilde}.

The infimum in~\eqref{def:locWM} is achieved, and the corresponding vector field $v$  is an element of the set~\cite[Th.~8.3.1]{AmbrosioGigliSavare08}
\begin{equation}
\label{set:tangents}
\overline{\{\nabla \varphi: \varphi\in C_c^\infty(\R^d)\}}^{L^2(\rho)}.
\end{equation}
The distribution $s$ and the corresponding velocity field $v$ are two different ways of describing the same tangent vector. Introducing the notation
\[
(s,v)\in \Tan{\rho} \qquad\Longleftrightarrow\qquad  s+\div \rho v =0 \quad\text{and}\quad
v\in \overline{\{\nabla \varphi: \varphi\in C_c^\infty(\R^d)\}}^{L^2(\rho)},
\]
this optimal vector field allows us to define a local metric tensor on the tangent space, the corresponding local inner product:
\begin{equation}
\label{def:Wip}
(s_1,s_2)_{-1,\rho} := \int_{\R^d} v_1\cdot v_2 \,d\rho 
\qquad\text{where } (s_i,v_i)\in \Tan{\rho}, 
\end{equation}
such that the norm~\eqref{def:locWM} satisfies $\|s\|_{-1,\rho}^2 = (s,s)_{-1,\rho}$. 

\medskip

A consequence of the construction above is that $\aK\D \aF$ can be viewed as the \emph{Wasserstein gradient} of $\aF$, in the sense that along a curve $t\mapsto \rho_t$,
\begin{equation}
\label{eq:Wasserstein-chain-rule}
\partial_t \aF(\rho_t) = \dual{\aF'(\rho_t)}{\dot \rho_t} = \bigl(\aK\D \aF(\rho_t),\dot\rho_t\bigr)_{-1,\rho_t}.
\end{equation}

\subsection{Formulation in terms of a process space}

The definitions above can be simply summarized in terms of the process-space concept introduced above, as follows:
\begin{align*}
&\text{State space:} & \aZ &= \M_2(\R^d)\\
&\text{Process space}: \quad & P_\rho\aZ &= \overline{\{\nabla \varphi: \varphi\in C_c^\infty(\R^d)\}}^{L^2(\rho)}\\[\jot]
&\text{Mapping $\aP$}: & \aP(\rho) w &= -\div \rho w\\[\jot]
&\text{Tangent space}: \quad& T_\rho\aZ &= \{\aP(\rho)w: w\in P_\rho\aZ\}=\{-\div \rho w: w\in P_\rho\aZ\} \\
&\text{Cotangent space:} & T^*_\rho\aZ &= \overline{C_c^\infty(\R^d)}^{\|\cdot\|_{1,\rho}}.
\end{align*}
The Wasserstein metric tensor is indeed a good example of how the process space gives a representation of the tangent space, through the operator $\aP$, in terms of which the dissipation is more directly characterized.

\subsection{The Wasserstein distance}
\label{sec:WassersteinDistance}

Based on the norm $\|\cdot\|_{-1,\rho}$, the \emph{Wasserstein distance} between two measures of equal mass can then be defined as the infimum of the Wasserstein norm integrated along curves~\cite{BenamouBrenier00},
\begin{equation}
\label{def:WassersteinDistance}
W_2(\rho_0,\rho_1)^2 = \inf\left\{
\int_0^1 \|\dot\rho_t \|^2_{-1,\rho_t}\, dt: \rho:[0,1]\to\M_2(\R^d), \ 
\rho_t\bigr|_{t=0,1} = \rho_{0,1}\right\}.
\end{equation}

The distance function $W_2$ is a complete separable metric on each subspace of $\M_2(\R^d)$ of functions of a given mass~\cite[Prop.~7.1.5]{AmbrosioGigliSavare08}; the distance between measures of different mass is not defined, since the condition $\dot \rho_t + \div \rho_t w_t =0$ implies that the total mass of $\rho_t$ is independent of $t$. Convergence in Wasserstein metric is equivalent to the combination of (a) narrow convergence of measures (i.e. in duality with continuous and bounded functions) and (b) convergence of the second moments~\cite[Prop.~7.1.5]{AmbrosioGigliSavare08}.
%
%

\subsection{The Wasserstein distance for empirical measures}

For the special case of empirical measures the Wasserstein norms and distance take a particularly simple form. This case also explains some of the relevance of this metric for modelling. 

Let $x_i:[0,1]\to\R^d$, $i=1,\dots n$ be $n$ time-courses of positions in $\R^d$, and define the time-parametrized empirical measure $\rho_t := \frac1n \sum_{i=1}^n \delta_{x_i(t)}$. Then 
the local Wasserstein norm is 
\begin{equation}
\label{eq:Wass-norm-atomic}
\|\dot \rho_t\|_{-1,\rho_t} ^2 = \frac1n \sum_{i=1}^n \dot x_i(t)^2.
\end{equation}
If $n=1$, for instance, then the Wasserstein norm is the same as the Euclidean norm of the velocity $\dot x$. For general $n$, the Wasserstein norm is proportional to the Euclidean norm of the velocity vector $(\dot x_1,\dots,\dot x_n)\in \R^{nd}$. As $n$ becomes large, the Wasserstein norm can be seen as a version of the Euclidean norm of the $n$-particle velocity that remains meaningful in the limit $n\to\infty$. 

Turning to the distance $W_2$, it follows from the discussion above that for single delta functions $W_2$ reproduces the Euclidean distance:
\[
W_2(\delta_{x},\delta_{y}) = |x-y|.
\]
For multiple delta functions the situation is similar; however, the infimum in the definition translates into a combinatorial optimization:
\[
W_2\Bigl(\frac1n\sum_{i=1}^n \delta_{x_i},\frac1n\sum_{j=1}^n \delta_{y_j}\Bigr)^2
= \inf_{\sigma\in S_n} \frac1n \sum_{i=1}^n |x_i-y_{\sigma(i)}|^2.
\]
In constructing curves connecting the two empirical measures, there is freedom in the matching between initial and final points; this matching results in the minimization over the set of permuations $S_n$ of $n$ objects.

%
%
%
%
%
%
%

\section{Comments}

\paragraph{Characterization of tangent and cotangent spaces.}
Our introduction of tangent and cotangent vectors implicitly  assumed that $\aZ$ is a space with a differentiable structure, in which time derivatives $\dot z$ and state-space derivatives $E'$ can be defined in a meaningful way. In addition, our implicit definition of the cotangent space as the topological dual of the tangent space requires a topology on the tangent space, which we didn't define. In many cases of importance, including the important example of the Wasserstein metric space,  these aspects are non-trivial, and appropriate constructions and generalizations are necessary; the book by Ambrosio, Gigli and Savar\'e~\cite{AmbrosioGigliSavare08} is an example of such generalizations in the case of gradient flows. 

\red{\paragraph{Historical notes.} Was Raleigh the first to write down this quadratic dissipation minimization problem? Talk about the name Wasserstein.}

\section{Exercises}

\begin{exercise}
In Section~\ref{sec:GF-intro} we made the claim that using tangents and cotangents makes the equations automatically coordinate-independent, as a physical theory should be. 

Verify this for a gradient flow
\[
\aZ = \R^2, \quad \aF:\aZ\to\R, \quad \aG(z)\in \R^{2\times 2},
\]
for instance by applying a linear transformation of the type
\[
y := A z \qquad\text{with} \qquad A\in \R^{2\times 2} \quad \text{invertible}.
\]
How do  the expressions for $\aF$ and $\aG$ in terms of $y$ relate to those in terms of $z$? And  $\aK$?
\end{exercise}

\begin{exercise}
In Section~\ref{sec:GF-intro} the claim was made that it is equivalent to choose any one of $\aG$, $\aK$, $\Psi$, or $\Psi^*$. A more precise statement would be the following. 

For given $\aG$ and $\aK$, define the norms and bilinear forms on $T\aZ$ and $T^*\aZ$ respectively:
\begin{alignat*}2
\|s\|_{\aG,z}^2 &\isdef  \dual{\aG(z)s}s \qquad\qquad& \|\xi\|_{\aK,z}^2 &\isdef \dual{\aK(z)\xi}\xi\\
(s_1,s_2)_{\aG,z} &\isdef \dual{\aG(z)s_1}{s_2} \qquad\qquad& (\xi_1,\xi_2)_{\aK,z}^2 &\isdef \dual{\aK(z)\xi_1}{\xi_2}.
\end{alignat*}
Prove that if any of the six objects $\aG$, $\aK$, $\|\cdot\|_{\aG,z}^2$, $\|\cdot\|_{\aK,z}^2$, $(\cdot,\cdot)_{\aG,z}$, or $(\cdot,\cdot)_{\aK,z}$ is chosen, then the other five follow automatically. (For this exercise a form of non-degeneracy can be assumed).
\end{exercise}

\begin{exercise}
The norms $\|\cdot\|_{\aG,z}$ and $\|\cdot\|_{\aK,z}$ defined above satisfy a duality inequality:
\[
\forall s\in T_z\aZ, \xi\in T^*_z\aZ: \quad 
\dual s\xi \leq \|s\|_{\aG,z}\|\xi\|_{\aK,z}.
\]
Equality happens if and only if $s = \lambda \aK(z)\xi \Longleftrightarrow \xi = \lambda^{-1}\aG(z) s$ with $\lambda>0$. Prove these statements. 
\end{exercise}

\red{More exercises?}

\chapter{Entropy and free energy, in stationary situations}
\label{ch:entropy-free-energy-stationary}

In this chapter I focus on the concepts of \emph{entropy}, \emph{free energy}, and \emph{large-deviation theory}, and their interconnections. The main modelling insights that I want to establish are 

\begin{example}
\begin{itemize}
\item[M1] \emph{Entropy is best understood as the rate functional of the empirical measure of a large number of identical stochastic particles; it describes the probability distribution of that system} (Section~\ref{subsec:first-interpretation});
\item[M2] \emph{Entropy arises from the indistinguishability of the particles} (Section~\ref{subsec:degeneracy});
\item [M3] \emph{Free energy arises from the tilting of a particle system by exchanging energy with a heat bath} (Section~\ref{subsec:free-energy}).
\end{itemize}

These insights imply the following modelling guidelines:
\begin{itemize}
\item[M4] Any \emph{concentration} $c$ (in moles per $\un {m^3}$) contributes an entropy $S(c)$ to the free energy of the system. For particles without interaction or with weak interaction the entropy is $S(\rho) = -R\int c\log (c/c_0)$ (Section~\ref{sec:modelling-free-energy}).
\item[M5] The free energy of a system with energy $E$ and concentrations $c_1,\dots,c_n$ is $E - \sum_i TS(c_i)$  (Section~\ref{sec:modelling-free-energy}).
\item[M6] 
\label{M6} The equilibrium states of such a system are the global minimizers of the free energy (Section~\ref{subsec:free-energy}).
\end{itemize}
\end{example}

\section{Entropy}
\label{subsec:entropy}

\textit{What is entropy?} This question has been asked an unimaginable number of times, and has received a wide variety of answers. Here I will not try to summarize the literature, but only mention that I personally like the treatments in~\cite{LiebYngvason97,Evans01}.

%

Instead we define one version of entropy,  the  \emph{relative} entropy of two probability measures. We assume we are working in a complete metrizable separable space $\Omega$ with a $\sigma$-algebra $\Sigma$ that contains the Borel sets (see Appendix~\ref{app:measure-theory} for the concepts used here).

\begin{definition}
\label{def:H}
Let $\mu,\nu\in \M(\Omega)$, with $\mu,\nu\geq0$. 
The \emph{relative entropy} of $\mu$ with respect to $\nu$ is 
\[
\H(\mu|\nu) := \begin{cases}
\ds\int_\Omega f\log f \, \d \nu \qquad &\text{if } \mu \ll \nu \text{ and }\ds f = \frac{\d\mu}{\d \nu} \\
+\infty & \text{otherwise}.
\end{cases}
\]
\end{definition}

A very useful special case will be the relative entropy with respect to the Lebesgue measure: for $\rho\in \P(\Omega)$, 
\begin{equation}
\label{def:Ent}
\Ent(\rho) := \H(\rho|\Lebesgue^d),
\end{equation}
where $\Lebesgue^d$ is the Lebesgue measure on $\R^d$.

Before we discuss the interpretation of this object in the next section, we first mention a number of properties.

\begin{theorem}
\begin{enumerate}
\item If $\mu(\Omega) = \nu(\Omega)$, then $\H(\mu|\nu)\geq 0$, and $\H(\mu|\nu)=0$ if and only if $\mu=\nu$;
\item If $\mu(\Omega) = \nu(\Omega)$, then $2\|\mu-\nu\|_{TV}^2\leq \H(\mu|\nu)$ ({Csisz\'ar-Kullback-Pinsker inequality});
\item \label{th:propsRelEnt:3}
$\H$ is invariant under transformations of the underlying space, i.e. if $\varphi:\Omega\to\Omega$ is a  one-to-one measurable mapping, and $\varphi_\#\mu$ and $\varphi_\#\nu$ are the \iindex{push-forwards} of $\mu$ and $\nu$, then $\H(\mu|\nu) = \H(\varphi_\#\mu|\varphi_\#\nu)$.
\end{enumerate}
\end{theorem}

\begin{proof}
The first part of the theorem can be understood by writing $\H$ as
\[
\H(\mu|\nu) = \int (f\log f - f + 1)\, d\nu\qquad\text{if }f = \frac{\d\mu}{\d\nu},
\]
and using the fact that $g(s) = s\log s - s+1$ is non-negative and only zero at $s=1$. For the Csisz\'ar-Kullback-Pinsker inequality we refer to~\cite{Csiszar67,Kullback67}. 

To prove the invariance under transformations, note that for all $\omega\in\Omega$,
\[
\frac{\d\varphi_\#\mu}{\d\varphi_\#\nu} (\varphi(\omega)) = \frac{\d\mu}{\d\nu}(\omega), 
\]
so that 
\begin{multline*}
\H(\mu|\nu) = \int \log\frac{\d\mu}{\d\nu}(\omega) \, \mu(d\omega)
= \int \log \frac{\d\varphi_\#\mu}{\d\varphi_\#\nu}(\varphi(\omega))\, \mu(d\omega)=\\
= \int \log \frac{\d\varphi_\#\mu}{\d\varphi_\#\nu}(\omega')\, \varphi_\#\mu(d\omega')
= \H(\varphi_\#\mu|\varphi_\#\nu).
\end{multline*}

\end{proof}

The properties in this theorem are relevant for the central role that this relative entropy plays. The non-negativity and Csisz\'ar-Kullback-Pinsker inequality show that when~$\mu$ and~$\nu$ have the same mass (e.g.\ if $\mu,\nu\in\P(\Omega)$) then $\H$ acts like a measure of distance between~$\mu$ and~$\nu$. It's not a distance function, since it is not symmetric ($\H(\mu|\nu)\not=\H(\nu|\mu)$), but via the Csisz\'ar-Kullback-Pinsker inequality it does generate the topology of total variation on the space of probability measures. 

The fact that $\H$ is invariant under transformations of space is essential for the modelling, for the following reason. Every modelling process involves a choice of coordinates, and this choice can often be made in many different ways. Nevertheless, if we believe that there is a well-defined energy that drives the evolution, then the value of this energy should not depend on which set of coordinates we have chosen to describe it in. 

Note that $\Ent$ is \emph{not} invariant under change of coordinates. This implies that this functional corresponds to a specific choice of coordinates. 

\begin{danger}
In part~\ref{th:propsRelEnt:3} of the theorem above, if $\varphi$ is a general mapping from $\Omega$ to a set $\Omega'$, then we still have an inequality: $\H(\varphi_\#\mu|\varphi_\#\nu)\leq \H(\mu|\nu)$. In fact, the discrepancy between the left- and right-hand sides of this inequality can be fully characterized by a \emph{tensorization identity}
\[
\H(\varphi_\#\mu|\varphi_\#\nu) =  \H(\mu|\nu) + \int_{\Omega'} \H(\mu_y|\nu_y) \, \varphi_\#\mu(dy).
\]
Here $\mu_y$ and $\nu_y$ are defined through the \iindex{disintegration theorem}~\cite[Th.~5.3.1]{AmbrosioGigliSavare08}.
\end{danger}

\section{Entropy as measure of degeneracy}
\label{subsec:degeneracy}

This mathematical definition of the relative entropy $\H$ above does not explain why it might appear in a model. 
One way to give an interpretation to the relative entropy $\H$ is by a counting argument, that we explain here for the case of a finite state space. It involves \emph{microstates} and \emph{macrostates}, where multiple microstates correspond to a single macrostate. The result will be that the entropy characterizes two concepts: one is the degree of degeneracy, that is the number of microstates that corresponds to a given macrostate, and the other is the probability of observing a macrostate, given a probability distribution over microstates. The two are closely connected. 

Take a finite state space~$I$ consisting of $|I|$ elements. If $\mu\in \P(I)$, then $\mu$ is described by $|I|$ numbers $\mu_i$, and the relative entropy with respect to~$\mu$ is 
\[
\H(\rho|\mu) = \sum_{i\in I} \rho_i\log\frac{\rho_i}{\mu_i}, \qquad \text{for }\rho\in \P(I).
\]

Consider $N$ particles on the lattice described by $I$, i.e. consider a mapping $x:\{1, \dots, N\}\to I$. We think of $x$ as the \emph{microstate}. Define an \emph{empirical measure} $\rho\in \P(I)$ by 
\begin{equation}
\label{def:ki}
k_i := \#\bigl\{j\in \{1,\dots,N\}: x(j) = i\bigr\} \quad\text{and}\quad
\rho_i := \frac {k_i}N, \qquad\text{for }i\in I.
\end{equation}
In going from $x$ to $\rho$ there is loss of information;  multiple mappings $x$ produce the same empirical measure $\rho$. The degree of degeneracy, the number of unique mappings $x$ that correspond to a given $\rho$, is $N!\,\bigl(\prod_{i\in I}k_i!\bigr)^{-1}$.   The fact that the particles are identical, \emph{indistinguishable}, is important here---this is required for the description in terms of integers~$k_i$. Because of this loss of information, we think of $\rho$ as the \emph{macrostate}.

We now determine the behaviour of this `degree of degeneracy' in the limit $N\to\infty$. Using Stirling's formula in the form
\[
\log n! = n\log n - n + o(n) \qquad\text{as } n\to\infty,
\]
we estimate
\begin{align*}
\log {N!} \, \bigl(\prod_{i\in I}k_i!\bigr)^{-1} &= \log N! - \sum_{i\in I} \log k_i!\\
&= N\log N - N - \sum_{i\in I} (k_i\log k_i - k_i) + o(N)\\
&= -N\sum_{i\in I} \rho_i \log \rho_i + o(N) \qquad\text{as }N\to\infty.
\end{align*}

One interpretation of the relative entropy therefore is as follows. Take for the moment $\mu_i$ to be the uniform measure, i.e. $\mu_i = |I|^{-1}$; then
\[
\H(\rho|\mu) = \sum_{i\in I} \rho_i\log{\rho_i} + \log|I|.
\]
Then
\begin{equation}
\label{char:H}
\H(\rho|\mu) = -\lim_{N\to\infty} \frac1N \log \#\text{realizations of }\rho
\ +\  \log |I|.
\end{equation}
This shows that if the number of microscopic realizations~$x$ of the macroscopic object~$\rho$ is large, then $\H(\rho|\mu)$ is small, and vice versa.\footnote{There is a problem here, though; the $\rho$ in the right-hand side can not be independent of $N$, since it is a vector with components of the form $k/N$, while the $\rho$ in the left-hand side  expression can not depend on $N$. This contradiction will be resolved when we discuss large deviations in the next section.} This is the interpretation in terms of a counting argument. 

\smallskip
We now switch to the probabilistic point of view.
If we allocate particles at random with the same, independent, probability for each microstate $x$, then the probability of obtaining each microstate is $|I|^{-N}$, and the probability of a macrostate $\rho$ satisfies
\[
\log \Prob(\rho) = \log |I|^{-N} N!\,\Bigl(\prod_{i\in I}k_i!\Bigr)^{-1}\\
=   - N \H (\rho|\mu) + o(N) \qquad\text{as }N\to \infty.
\]

We can do the same with non-equal probabilties: we place each particle at an $i\in I$ with probability $\mu_i$. Then the probability of a microstate $x$ is 
\[
\prod_{j=1}^N \mu_{x(j)},
\]
and now
the probability of a macrostate $\rho$  satisfies
\begin{align*}
\log \Prob(\rho) &= \log \Bigl(\prod_{j=1}^N \mu_{x(j)} \Bigr) N!\,\Bigl(\prod_{i\in I}k_i!\Bigr)^{-1}\\
&= \sum_{j=1}^N \log \mu_{x(j)} +  \log {N!} \, \bigl(\prod_{i\in I}k_i!\bigr)^{-1}\\
&= N\sum_{i\in I} \rho_i \log\mu_i - N \sum_{i\in I} \rho_i\log \rho_i + o(N)\\
&= -N\H(\rho|\mu) + o(N) \qquad\qquad\text{as }N\to \infty.
\end{align*}



The common element in both points of view is the degeneracy, the number of microstates that is mapped to a single macrostate. Of course, this degeneracy only arises if the particles can not be distinguished from each other. Therefore 
I like to summarize this section like this:
\begin{example}
{Entropy arises from the indistinguishability of the particles in an empirical measure}.
\end{example}
We will return to this issue in Section~\ref{subsec:first-interpretation}.

\section{Degeneracy and dynamics}
\label{sec:degeneracy-dynamics}
The driving functional of Section~\ref{sec:diffusion-particles-fluid}, $\aF: \aZ\to\R$, can be written as
\begin{equation}
\aF:=RT\Ent+\mathcal{E}_G,\qquad 
\mathcal{E}_G= \rho g\int_\Omega x_3 c(x) \, dx,\label{eq:Ent_Pot-Energy}
\end{equation}
where $\Ent(\rho) = \int\rho\log\rho $ as in~\eqref{def:Ent}.
Our understanding of the physics suggests that there must be a close relationship between the energy of the system (modelled in terms of  $\aF$) and the forces acting on the system. The  calculation 
\begin{align*}
-\nabla \frac{\delta\mathcal{E}_G}{\delta c}=-\nabla\left(\rho g x_3\right)=-\rho g \vec{e}_3,
\end{align*}
where  $\vec{e}_3$ is the unit vector in the vertical direction, clearly suggests that $\mathcal{E}_G$ generates the gravitational force.
A similar computation for $\Ent$ yields 
\begin{align}
-\nabla \frac{\delta \Ent}{\delta c}=-\nabla RT\left(
\log\frac{c}{c_0}+1\right)=-RT\frac{\nabla c}{c}. \label{eq:Formal-Der-Ent}
\end{align}
A casual glance at \eqref{eq:Formal-Der-Ent} suggests the following natural questions:
\begin{enumerate}
\item{Why does $\Ent$ have this particular formulation and  does it generate a force?}
\item{Why is diffusion temperature-dependent, and particularly, why is this dependence linear in $T$? Why is the \textit{universal gas constant} $R$ involved?}
\item{Why do the properties  of the  solvent seem to play no role?}
\end{enumerate}
 
\paragraph{Answer to question 1:} Entropy indeed generates a force. This force arises from thermal agitation by the surrounding medium; the temperature $T$ is in fact the temperature of this medium. At finite particle number $N$ this is visible as stochastic forces. In limit $N\rightarrow\infty$ an `entropic force' is what remains.

Let us illustrate this with a simple example.
Consider a collection of bins (see Fig.~\ref{fig:Bin-Model}) with `sites' in them. We think of the sites as \emph{microstates}, and the bins as \emph{macrostates}; many microstates are bundled into a single macrostate, but the number of microstates may vary (and does vary). We let a particle jump from one site to another, that is, from one microstate to another, at random; the only restriction is that the sites have to belong to adjacent bins, i.e., a particle can only jump to a site in a bin adjacent to the current bin.

\begin{figure}[h]
\centering
\includegraphics[scale=0.55]{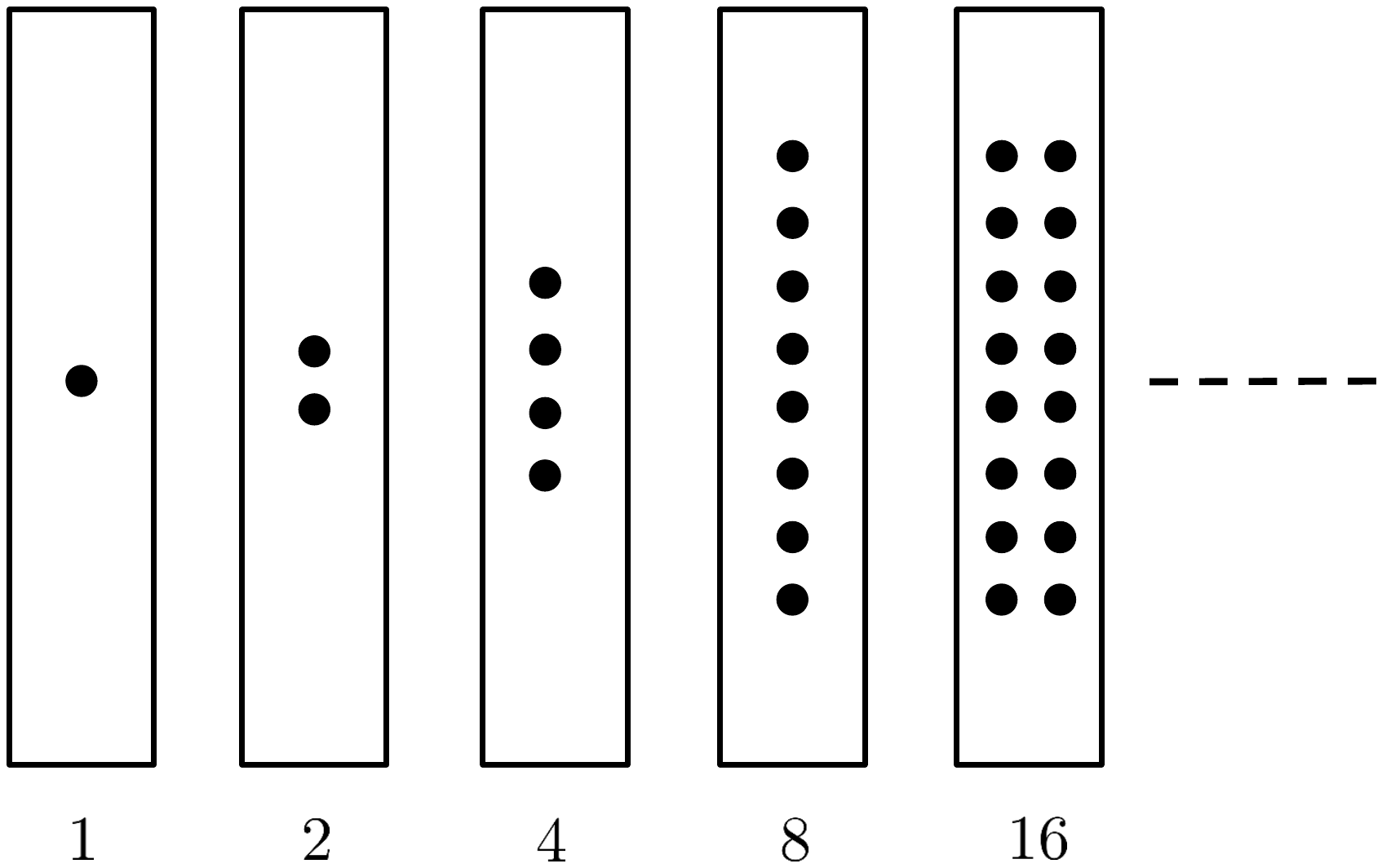}
\caption[BinModel]{Here each bin represents a macrostate and and a particles in the bin represents a microstate corresponding to the associated microstate.}
\label{fig:Bin-Model}
\end{figure}
 
Although the selection of the next site is uniform over the accessible sites, the difference in numbers of sites-per-bin creates a net flux towards the bins with many sites, i.e.\ to the right in Figure~\ref{fig:Bin-Model}. If the random motion of the particles is achieved by some force, then the net flux should correspond to a net force.
This, now, is the interpretation of an `entropic' force.

\begin{example}
An \emph{entropic force} is a force that arises from unbiased microscopic motion in combination with variation in numbers of sites per macrostate. 
\end{example}
As we mentioned above, the quantity `number of sites per macrostate' is often referred to as the \emph{degeneracy} of a macrostate. Note that what counts is not the absolute number of microstates, nor the difference in number, but the \emph{ratio}: in Figure~\ref{fig:Bin-Model} the net flux and force to the right are the same for each bin, because the ratio of sites is the same for each pair of adjacent bins. This explains why the driving force is the spatial derivative of the \emph{logarithm} of the degeneracy, since that measures exactly the variation in ratio from one point in space to the next. 

Note that the calculation in the previous section explains why for empirical measures (i.e. particles) the logarithm of the measure of degeneracy is proportional to $\rho\log\rho$.

\paragraph{Answer to question 2:} On average, each solvent molecule has the same kinetic energy $\frac{1}{2}kT$ in each direction, where $k=$ converts our temperature scale (Kelvin) into our energy scale (Joule). The gas constant $R=k N_A$ is essentially the same as $k$, but measured per mole instead of per molecule. The kinetic energy increases with temperature, and therefore the collisional forces between molecules do the same.\footnote{To be honest, I don't have a good explanation why the diffusion constant scales linearly with temperature. Einstein~\cite{Einstein05} derives this from Gibbs' law of thermodynamics, but that seems overpowered for the purpose. Suggestions are welcome.}

\paragraph{Answer to question 3:} The properties of the solvent do play a role. As we shall see in Chapter~\ref{ch:Wasserstein-dissipation}, the prefactor of the entropy, i.e.\ $RT$, combines with the viscosity to create the diffusion constant. In the entropy, however, dynamics is not taken into account, and only the degeneracy plays a role---and the degeneracy is independent of the solvent. 

\section{Large deviations}
\label{subsec:ldp}

We now turn to a related interpretation of entropy, and especially relative entropy. 
For these purposes, the main role of the relative entropy is in the characterization of \emph{large deviations of empirical measures}. 

Large deviations are best explained by an example. We toss a balanced coin $n$ times, and we call $S_n$ the number of heads. Well-known properties of $S_n$ are\footnote{This section draws heavily from the introduction in~\cite{DenHollander00}.}
\begin{itemize}
\item $\ds \frac1n S_n\stackrel{n\to\infty} \longrightarrow \frac12 $ almost surely (the law of large numbers)
\item $\ds\frac2{\sqrt n} \bigl(S_n-\frac n2\bigr) \stackrel{n\to\infty} \longrightarrow Z$ in law, where  $Z$ is a standard normal random variable (the central limit theorem).
\end{itemize}

The second (which contains the first) states that $S_n$ is typically $n/2$ plus a random deviation of order $O(1/\sqrt n)$. Deviations of this size from the expectation are called normal. Large deviations are those that are larger than normal, such as for instance the event that $S_n\geq an$ for some $a>1/2$. Such large-deviation events have a probability that vanishes as $n\to\infty$, and a \emph{large-deviation principle} characterizes exactly how fast it vanishes. A typical example is
\begin{equation}
\label{ldp:1}
\text{For any }a\geq \frac12, \qquad
\Prob(S_n\geq an) \sim e^{-nI(a)} \qquad\text{as }n\to\infty,
\end{equation}
where
\[
I(a) := \begin{cases}
a\log a+ (1-a)\log (1-a) + \log 2 & \text{if } 0\leq a\leq 1\\
+\infty & \text{otherwise}
\end{cases}
\]
The characterization~\eqref{ldp:1} states that the probability of such a rare event decays exponentially in $n$, for large $n$, for each $a\geq1/2$. 
The function $I$ is called the \emph{rate function}, since it characterizes the constant in the exponential decay. 
\begin{figure}[ht]
\centering
\noindent
\psfig{height=4cm,figure=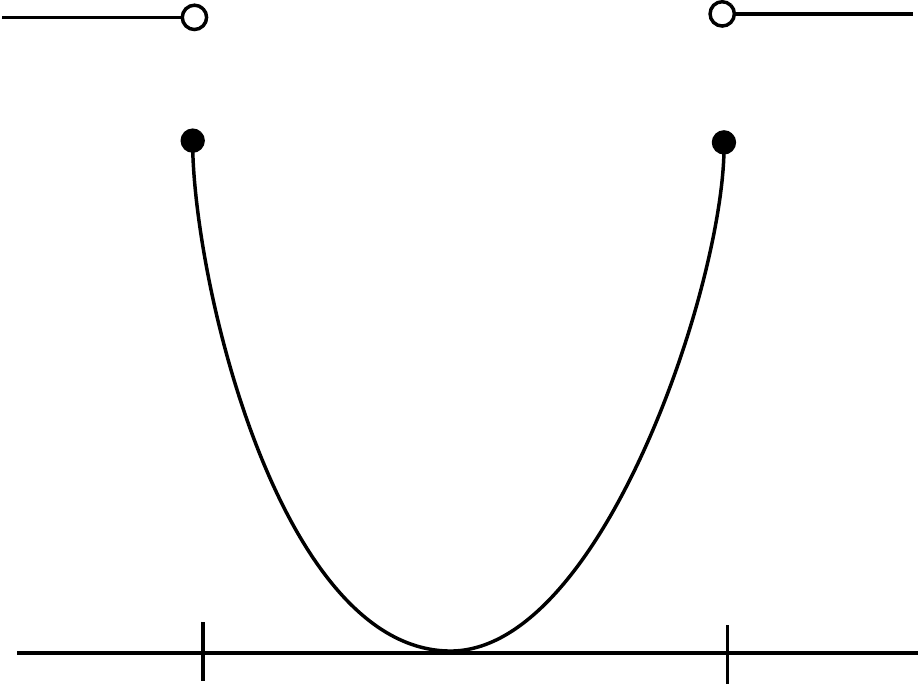}
\caption{The function $I$}
\end{figure}

In order to explain what exactly the symbol $\sim$ in~\eqref{ldp:1} means we give a precise definition of a large-deviation principle.  

\begin{definition}
\label{def:ldp}
A sequence $\mu_n\in \P(\Omega)$ satisfies a \emph{large-deviation principle} with speed $n$ and rate function $I$ iff
\begin{alignat*}2
&\forall O\subset \Omega \text{ open}, &\qquad & \liminf_{n\to\infty}\frac1n\log \mu_n(O)\geq -\inf_{O} I\\
&\forall C\subset \Omega \text{ closed}, &\qquad & \limsup_{n\to\infty} \frac1n \log \mu_n(C) \leq -\inf_{C} I.
\end{alignat*}
\end{definition}

Let us make a few remarks. 
\begin{itemize}
\item The definition of the large-deviation principle is rather cumbersome. We often write it, formally, as
\[
\Prob(X_n\approx x) \sim e^{-nI(x)},
\]
which is intended to mean exactly the same as Definition~\ref{def:ldp}.
\item The characterization~\eqref{ldp:1} can be deduced from Definition~\ref{def:ldp} as follows. Take $\Omega = \R$, the set in which $\frac1n S_n$ takes its values, and define
\[
\mu_n(A) := \Prob \Big(\frac1n S_n\in A\Big).
\]
Then $\Prob(S_n\geq an) = \mu_n([a,\infty))$, and note that $\inf_{[a,\infty)} I = I(a)$ whenever $a\geq1/2$. Supposing that the large-deviation principle has been proved for $\mu_n$ with rate function~$I$ (see e.g.~\cite[Th.~I.3]{DenHollander00}), we then find that for all $1/2\leq a<1$
\begin{align*}
-I(a) = -\inf_{(a,\infty)}\!\!I &\;\leq\; \liminf_{n\to\infty} \Prob(S_n>an)\\
&\;\leq\; \limsup_{n\to\infty} \Prob(S_n\geq an)
\;\leq\; -\inf_{[a,\infty)}\!\!I \;= \;-I(a).
\end{align*}
Therefore, if $1/2\leq a<1$ then the liminf and limsup coincide, and we have that 
\[
\lim_{n\to\infty} \frac1n \log \Prob(S_n\geq an)
=\lim_{n\to\infty} \frac1n \log \Prob(S_n> an)
= -I(a).
\]
This is the precise version of~\eqref{ldp:1}. 
\item In the example of the coin, the two inequalities in Definition~\ref{def:ldp} reduced to one and the same equality at all points between $0$ and $1$, by the continuity of $I$ at those points. In general a rate function need not be continuous, as the example of $I$ above shows; neither is the function $I$ unique. However, we can always assume that $I$ is lower semi-continuous, and this condition makes $I$ unique.
\item Looking back at the discussion in Section~\ref{subsec:degeneracy}, we see that for instance the limit~\eqref{char:H} is a large-deviation description, at least formally. 
We also remarked there that the characterization~\eqref{char:H} can not be true as it stands. In Definition~\ref{def:ldp} we see how this is remedied: instead of a single macrostate $\rho$, we consider open and closed sets of macrostates, which may contain states $\rho$  of the form $k/N$ for different values of~$N$.
\end{itemize}

\begin{remark}
The definition of the large-deviation principle has close ties to two other concepts of convergence. 
\begin{itemize}
\item A sequence of probability measures $\mu_n$ converges \emph{narrowly} (in duality with continuous and bounded functions) to $\mu$ if 
\begin{alignat*}2
&\forall O\subset \Omega \text{ open}, &\qquad & \liminf_{n\to\infty}\mu_n(O)\geq \mu(O)\\
&\forall C\subset \Omega \text{ closed}, &\qquad & \limsup_{n\to\infty}  \mu_n(C) \leq \mu(C).
\end{alignat*}
Apparently, the large-deviation principle corresponds to a statement like `the measures $\frac1n \log\mu_n$ converge narrowly'.

\item The definition in terms of two inequalities also recalls the definition of Gamma-convergence, and indeed we have the equivalence 
\[
\mu_n \text{ satisfies a large-deviation principle with rate function $I$} 
\quad\Longleftrightarrow\quad
\frac1n \H(\,\cdot\,|\mu_n) \stackrel\Gamma\longrightarrow \hat I,
\]
where $\hat I(\nu) := \int I\,d\nu.$ A proof is given in~\cite{Mariani12}.

\end{itemize}

\end{remark}

A property of large deviations that will come back later is the following. Suppose that we have a large-deviation result for a sequence of probability measures $\mu_n$ on a space $\X$ with rate functional $I$. Suppose that we now \emph{tilt} the probability distribution $\mu_n$ by a functional $F:\X\to\R$, by defining the new measure
\[
\tilde \mu_n(A) = \frac{\ds\int_A e^{-nF(x)}\,\mu_n(dx)}{\ds\int_\X e^{-n F(x)}\, \mu_n(dx)}.
\]
This increases the probability of events $x$ with lower $F(x)$, with respect to events $x$ with higher $F(x)$, with a similar exponential rate (the prefactor $n$) as a large-deviation result. 

The large-deviation behaviour of $\tilde \mu_n$ is now given by
\begin{theorem}[Varadhan's Lemma (e.g.~{\cite[Th.~III.13]{DenHollander00}})]
\label{th:VaradhansLemma}
Let $F:\X\to\R$  be continuous and bounded from below. Then $\tilde\mu_n$ satisfies a large deviation principle with rate function 
\[
\tilde I(x) := I(x) + F(x) - \inf_{\X}(I+F).
\]
\end{theorem}
The final term in this expression is only a normalization constant that makes sure that $\inf \tilde I = 0$. The important part is that $\tilde I$ is the sum of the two functions $I$ and $F$.
In words: if we modify a probability distribution by {tilting} it with an exponential factor $e^{-nF}$, then that tilting function $F$ ends up being \emph{added} to the original rate function $I$.

\section{Entropy as large-deviation rate function}
\label{subsec:first-interpretation}

Now to the question why relative entropy appears in the context of thermodynamics. Consider the following situation. We place $n$ independent particles in  a space $\X$ according to a distribution $\mu\in\P(\X)$, i.e.\ the probability that the particle is placed in a set $A\subset\X$ is $\mu(A)$. We now consider the {empirical measure} of these $n$ particles, which is the measure
\[
\rho_n := \frac1n \sum_{i=1}^n \delta_{X_i},
\]
where $X_i$ is the position of the $i^{\mathrm{th}}$ particle. The empirical measure $\rho_n$ is a random element of $\P(\X)$, and the law of large numbers gives us that with probability one $\rho_n$ converges weakly (in the sense of measures) to the law $\mu$. (This is of course the standard way of determining~$\mu$ if one only has access to the sample points $X_i$).

In this situation the large deviations of $\rho_n$ are given by Sanov's theorem (see e.g.~\cite[Th.~6.2.10]{DemboZeitouni98}). The random measure $\rho_n$ satisfies a large-deviation principle with rate $n$ and rate function 
\[
I(\rho) := \H(\rho|\mu),
\]
or in the shorthand notation that we used earlier,
\begin{equation}
\label{char:ldp-Sanov}
\fbox{$\ds\Prob(\rho_n\approx \rho) \sim e^{-n\H(\rho|\mu)} \qquad \text{as } n\to\infty.$}
\end{equation}
This is such an important result that I state it separately:
\label{quote:H}
\begin{example}
\noindent{The relative entropy is the rate functional of the empirical measure of a  large number of identical particles.}
\end{example}
Note the stress on `empirical measure': the appearance of the relative entropy is intimately linked to the fact that we are considering empirical measures. Section~\ref{subsec:degeneracy} gives an insight into why this is: when passing from a vector of positions to the corresponding empirical measure, there is loss of information, since  particles at the same position are indistinguishable.

\section{Free energy and the Boltzmann distribution}
\label{subsec:free-energy}

In many books one encounters in various forms the following claim. Take a system of particles living in a space $\X$, and introduce an `energy' $E:\X\to\R$ (in Joules, $\un J$) depending on the position $x\in \X$. Bring the system of particles into contact with a `heat bath' of temperature $T$ (in Kelvin, $\un K$), and let it settle into equilibrium. Then the probability distribution of the particles will be given by 
\begin{equation}
\label{prob:Boltzmann}
\Prob(A) = \frac{\ds\int_A e^{-E(x)/kT}\, dx}{\ds \int_{\X} e^{-E(x)/kT}\,dx}.
\end{equation}
This is known as the \emph{Boltzmann distribution}, or \emph{Boltzmann statistics}, and the \emph{Boltzmann constant} $k$ has the value $1.4 \cdot 10^{-23}\un{J/K}$. (Note that it only exists if the exponentials are integrable, which is equivalent to sufficient growth of $E$ for large $x$. We will assume this for this discussion).

\medskip
Where does this distribution come from? The concept of entropy turns out to give us the answer. 

Since we need a system and a heat bath, we take two systems, called $S$ and $S_B$ (for `bath'). Both are probabilistic systems of particles; $S$ consists of $n$ independent particles $X_i\in\X$, with probability law $\mu\in\P(\X)$; similarly $S_B$ consists of $m$ independent particles $Y_j\in\Y$, with law $\nu\in\P(\Y)$. The total state space of the system is therefore $\X^n\times\Y^m$.

The \emph{coupling} between these systems will be via an \emph{energy constraint}. We assume that there are energy functions $e:\X\to\R$ and $e_B:\Y\to\R$, and we will constrain the joint system to be in a state of fixed total energy, i.e. we will only allow states in $\X^n\times\Y^m$ that satisfy 
\begin{equation}
\label{eq:energy-constraint}
\sum_{i=1}^n e(X_i) + \sum_{j=1}^m e_B(Y_j) = \text{constant}.
\end{equation}
The physical interpretation of this is that energy (in the form of heat) may flow freely from one system to the other, but no other form of interaction is allowed. 

Similar to the example above, we describe the total states of systems $S$ and $S_B$ by empirical measures $\rho_n = \frac1n \sum_i \delta_{X_i}$ and $\zeta_m = \frac1m\sum_j \delta_{Y_j}$. We define the average energies $E(\rho_n) := \frac1n\sum_i e(X_i) = \int_{\X}e\, d\rho_n$ and $E_B(\zeta_m) := \int_{\Y} e_B\,d\zeta_m$, so that the energy constraint above reads $nE(\rho_n) + mE_B(\zeta_m) = \text{constant}$.  

By Section~\ref{subsec:first-interpretation}, each of the systems \emph{separately} satisfies a large-deviation principle with rate functions $I(\rho) = \H(\rho|\mu)$ and $I_B(\zeta) = \H(\zeta|\nu)$. However, instead of using the explicit formula for $I_B$, we are going to assume that $I_B$ can be written as a function of the energy $E_B$ of the heat bath alone, i.e. $I_B(\zeta) = \tilde I_B(E_B(\zeta))$. For the coupled system we derive a joint large-deviation principle by choosing that (a) $m = nN$ for some large $N>0$, and (b) the constant in~\eqref{eq:energy-constraint} scales as $n$, i.e. 
\[
nE(\rho_n) + nNE_B(\zeta_{nN}) = n\overline E \qquad\text{for some }\overline E.
\]
The joint system satisfies then a large-deviation principle\footnote{This statement is formal; I haven't yet worked out how to formulate this  rigorously for general state spaces.}
\[
\Prob\Big((\rho_n,\zeta_{nN}) \approx (\rho,\zeta) \;\Big|\; E(\rho_n) + NE_B(\zeta_{nN}) = \overline E\Bigr) \sim \exp \bigl(- nJ(\rho,\zeta)),
\]
with rate functional
\[
J(\rho,\zeta) := 
\begin{cases}
\H(\rho|\mu) +N \tilde I_B(E_B(\zeta))+\text{constant}& \text{if } E(\rho) + NE_B(\zeta) = \overline E,\\
+\infty&\text{otherwise.}
\end{cases}
\]
Here the constant is chosen to ensure that $\inf J=0$.

The functional $J$ can be reduced to a functional of $\rho$ alone,
\[
J(\rho) = \H(\rho|\mu) + N\tilde I_B\left(\frac{\overline E - E(\rho)}N\right)+\text{constant}.
\]
In the limit of large $N$, one might approximate 
\[
N\tilde I_B\left(\frac{\overline E - E(\rho)}N\right) \approx N\tilde I_B(\overline E) - \tilde I_B'(\overline E)E(\rho).
\]
The first term above is absorbed in the constant, and we find
\[
J(\rho) \approx \H(\rho|\mu) -E(\rho)\tilde I_B'(\overline E) +\text{constant}.
\]
In many cases $I_B'$ is negative, since larger energies typically lead to higher probabilities and therefore smaller values of $I_B$. 
Now we simply define $kT := -1/\tilde I_B'(\overline E)$, and we find 
\[
J(\rho) \approx \H(\rho|\mu)+\frac1{kT} E(\rho) +\text{constant}.
\]
Compare this to the typical expression for free energy  $E-TS$; if we interpret $S$ as  $-k\H(\cdot|\mu)$, then we find the expression above, up to a factor $kT$. 

Note that the right-hand side can be written as $\H(\rho|\tilde \mu)$, where $\tilde \mu$ is the tilted distribution
\[
\tilde \mu(A) = \frac{\ds\int_A e^{-e(x)/kT}\, \mu (dx)}{\ds\int_{\X} e^{-e(x)/kT}\, \mu (dx)}
\]
Here we recognize the expression~\eqref{prob:Boltzmann} for the case when $\mu$ is the Lebesgue measure.

This derivation  shows that the effect of the heat bath is to \emph{tilt} the system $S$: a state $\rho$ of~$S$ with larger energy $E(\rho)$  implies a smaller energy $E_B$ of $S_B$, which in turn reduces the probability of~$\rho$. The role of temperature $T$ is that of an exchange rate, since it characterizes the change in probability (as measured by the rate function $I_B$) per unit of energy. When $T$ is large, the exchange rate is low, and then larger energies incur only a small probabilistic penalty. When temperature is low, then higher energies are very expensive, and therefore more rare.
From this point of view, the Boltzmann constant $k$ is simply the conversion factor that converts our Kelvin temperature scale for $T$ into the appropriate `exchange rate' scale.

\medskip
One consequence of the discussion above is that the expression 
\begin{equation}
\label{exp:free-energy-in-derivation-free-energy}
\aF(\rho) := \H(\rho|\mu)+\frac1{kT} E(\rho) +\text{constant}
\end{equation}
is the rate function for a system of particles in contact with a heat bath. Let us make this statement precise, because this will explain some of the questions that we started this chapter with.  

We recall the experiment that we just formulated. For each $n$, draw $n\times m$ particle positions $\{X^n_i\}{i=1}^n$ and $\{Y^n_j\}_{j=1}^{nN}$ from the joint state space $\X^n\times\Y^{nN}$, with distribution $\mu$ (for $X^n_i$) and $\nu$ (for $Y^n_j$), conditioned on the equality~\eqref{eq:energy-constraint} with constant equal to $n\overline E$. Then, writing again $\rho_n := n^{-1}\sum_i \delta_{X^n_i}$, we have seen that
\[
\Prob(\rho_n \approx \rho) \sim \exp(-n\aF(\rho))\qquad \text{as }n\to\infty.
\]
For large $n$, this characterization implies that the distribution of $\rho_n$ is strongly concentrated around the global minimizers of $\aF$---and in the limit $n\to\infty$, it collapses onto this set. This explains `modelling insight [M6]' on page~\pageref{M6}, \emph{the equilibrium states of this system are global minimizers of the free energy}.

\medskip

\begin{danger}
In books on thermodynamics one often encounters the identity (or definition) $T = dS/dE$. This is formally the same as our definition of $kT$ as $-dI_B/dE$, if one interprets $I_B$ as an entropy and adopts the convention to multiply the non-dimensional quantity $I_B$ with $-k$. 

Yet another insight that this calculation gives is the following. It might be puzzling that in defining a free energy (e.g.\ $E-TS$, or $kT \Ent +E$), one adds two rather different objects: the energy of a system seems to be a completely different type of object than the entropy. This derivation of Boltzmann statistics shows that  it's not exactly energy and entropy that one's adding; it really is more like adding two entropies ($\H$ and $I_B$, in the notation above). The fact that we write the second entropy as a constant times energy follows from the coupling and the approximation allowed by the assumption of a large heat bath.
%

\end{danger}

\section{Generalizations}
\label{sec:Free-energy-generalizations}

\paragraph{Interaction energy}

In fact, the argument in the previous section applies just as well with an energy $E$ that is not a simple sum over particles, but includes interaction effects. A typical example would be a pair-interaction energy of the type
\[
\frac1n \sum_{i=1}^n \Vb(X_i) + \frac1{n^2}\sum_{i,j=1}^n \Vi(X_i-X_j) = 
\int_\X \Vb\, d\rho_n + \iint_{\X\times\X} \Vi(x-y)\, \rho_n(dx)\rho_n(dy).
\]
Both electrostatic interaction energies and gravitaional energies are of this type. We come back to this type of energy in Section~\ref{sec:SDE}.

\paragraph{Multiple species}

The case of multiple species of particle can be treated by assuming that each particle has not only a position $X_i\in \Omega$ but also a \emph{type} $T_i\in\{1,\dots,m\}$. This corresponds to  taking the set $\X$ above as $\X:= \Omega\times \{1,\dots,m\}$. We define the empirical measure $\rho_n$ similarly as $\frac1n\sum_i \delta_{(X_i,T_i)}$, which is now a probability measure on the space $\X$ with variables $x\in\Omega$ and $t\in\{1,\dots,m\}$.

Following the same line of reasoning, we again find an approximate rate functional
\[
\H(\rho|\mu) + \frac1{kT} E(\rho) + \text{constant},
\]
where $\mu$ is now the original distribution of the particles on the state space $\X$. If we write 
\[
\rho_i := \rho\bigr|_{\{t = i\}}, \qquad \mu_i := \mu\bigr|_{\{t = i\}},
\]
then 
\[
\H(\rho|\mu) = \sum_i \H(\rho_i|\mu_i),
\]
so that the rate functional becomes
\begin{equation}
\label{exp:free-energy-multiple-species-dimensionless}
\sum_i \H(\rho_i|\mu_i) + \frac1{kT} E(\rho) + \text{constant}.
\end{equation}

\section{Dimensional versions}
\label{sec:dimensional-versions-free-energy}

In applications we work with concentrations that have dimensions, typically of moles per $\un{m^3}$ or per liter, and traditionally free energies have dimensions of energy rather than being dimensionless. To convert a dimensionless version such as~\eqref{exp:free-energy-multiple-species-dimensionless} into a dimensional form we first multiply with an energy quantity to make the dimension that of energy. The traditional choice is to multiply with $kT$. \red{Is `tradition' all there is to it?}

We next connect the dimensionless measures $\rho_i$ with dimensional concentrations $c_i$ in a physical domain $\Omega$ by 
\begin{equation}
\label{rel:rho-c}
\rho_i = \overline c^{-1} \,c_i\, \Lebesgue^3\bigr|_\Omega \qquad\text{with}\qquad \overline c := \sum_i \int_\Omega c_i(x)\, dx.
\end{equation}
The natural choice for $\mu$ in~\eqref{exp:free-energy-multiple-species-dimensionless} is to be proportional to the Lebesgue measure,
\[
\mu = \mu_0\Lebesgue^3\bigr|_\Omega \qquad\text{for some }\mu_0>0.
\]
This corresponds to assuming no preference of the particles for one place above another---apart from that which might be encoded in the energy $E$.
We then calculate that 
\[
kT\cdot \eqref{exp:free-energy-multiple-species-dimensionless} \;=\; 
kT\sum_i \overline c^{-1} \int_\Omega c_i(x)\log \frac{c_i(x)}{\overline c\mu_0}\, dx
 \;+\; E(c_1,\dots,c_m) \;+\; \text{constant},
\]
where we leave the dependence of $E$ on $(c_1,\dots,c_m)$ unspecified for the moment.

This expression can be thought of as the free energy \emph{per particle}: the energy $E$ in Section~\ref{subsec:free-energy} is defined as an energy per particle, and the entropy terms are invariant under changes in particle number. The total free energy is this expression multiplied by the total number of particles $n\overline c$. If $c$ is indeed measured in moles per cubic meter, then 
$n$ is known exactly to be \emph{Avogadro's number}, $N_A\simeq 6.022\cdot 10^{23}$. Multiplying by $N_A\overline c$ we find
\[
RT\sum_i  \int_\Omega c_i(x)\log \frac{c_i(x)}{c_0}\, dx
 \;+\; \tilde E(c_1,\dots,c_m) \;+\; \text{constant}.
\]
Here the constant $R =kN_A \simeq 8.314 \,\un{J}\un K^{-1}\un{mol}^{-1}$ is the same \iindex{universal gas constant} as in Section~\ref{sec:diffusion-particles-fluid}. In this expression we have replaced the energy-per-particle $E$ by the total energy $\tilde E$, and we replaced $\overline c\mu_0$ by a reference concentration $c_0$. Note that the value of $c_0$ only changes the value of the constant in this expression, and therefore we can choose it arbitrarily.

\section{Modelling free energy}
\label{sec:modelling-free-energy}

Let us now summarize. The discussions above give an algorithm how to construct the total free energy in systems described by concentrations. (For completeness: the total free energy is equal to the large-deviation rate functional, multiplied by $kT$ and by the total number of particles.) Take a system with $m$ concentrations $c_1,\dots,c_m:\Omega\to[0,\infty)$ (in moles per $\un{m^3}$), which can exchange energy with a heat bath with temperature $T$, assuming that we know the total energy $E(c_1,\dots,c_m)$ as a function of $(c_1,\dots,c_m)$. Apart from the interaction described by the energy, the particles are assumed to have no interaction with each other. By the arguments above we find that the free energy of this system is
\begin{equation}
\label{def:free-energy-modelling-entropy}
\aF(c_1,\dots,c_m) := E(c_1,\dots,c_m) + {RT}\sum_{i=1}^m \int_\Omega c_i(x)\log \frac{c_i(x)}{c_0}\, dx +\text{constant}.
\end{equation}

We saw in the sections above how the entropy terms in~\eqref{def:free-energy-modelling-entropy} arise from the indistinguishability of the particles of the same species: each species contributes a term $RT\int c\log c/c_0 $ to the free energy. 

\medskip

While modelling the system of diffusing particles, in Section~\ref{sec:diffusion-particles-fluid}, we chose a functional of this type,~\eqref{choice:entropy-diffusion-particles-fluid}. We now recognize that this is a free energy, i.e. a dimensional version of the large-deviation rate functional associated with concentrations of independent particles. This argument gives an interpretation of this functional, but it does not yet explain why it should be a driving force in a gradient flow. That aspect we discuss in the next chapter. 

%
%
%
%
%

\chapter{Free energy dissipated through Wasserstein metrics}
\label{ch:Wasserstein-dissipation}

\red{Think about changing this chapter: making it more academic}

In this chapter I want to establish the following modelling insights.

\begin{example}
\begin{itemize}
\item[M7] Particles moving through a  viscous fluid dissipate energy. If the starting and ending positions of $n$ particles are characterized by two empirical measures $\rho_0$ and $\rho_1$, then the minimally dissipated energy in a stationary Newtonian fluid is 
\[
\frac{n\eta} \tau W_2(\rho_0,\rho_1)^2,
\]
where $\tau$ is the duration of the motion, $\eta$ is a friction parameter, and $W_2$ is the Wasserstein distance.
\item [M8] The Wasserstein distance $W_2$ also characterizes the mobility of empirical measures of Brownian particles, in the sense of large deviations.

\item [M9] In an SDE with Brownian noise, in the many-particle limit, free energy is dissipated through the Wasserstein distance. Therefore in many cases the many-particle limit is a Wasserstein gradient flow of the free energy. 
\end{itemize}
Taken together, these will explain the modelling choices of Section~\ref{sec:diffusion-particles-fluid}.
\end{example}

\red{This doesn't explain the $D$.}

\section{Dissipation in a viscous fluid}

When a spherical particle is dragged through a Newtonian viscous fluid, the relative velocity~$v$ is linearly related to the force $f \,[\un N]$ that is required to keep it moving:
\begin{equation}
\label{eq:linear-drag-law}
f = \eta v.
\end{equation}
(Stokes' law  gives an explicit expression for the \emph{friction coefficient} $\eta$: $\eta = 6\pi \nu r$, where $r$ is the radius of the particle and $\nu\, [\un {Ns/m^2}]$ the dynamic viscosity  of the fluid; see~\cite[Eq.~(126)]{Stokes51} or~\cite[\textsection20]{LandauLifshitzVI87}.)
Therefore, if $X_1(t),\dots , X_n(t)$ are trajectories of $n$ particles, for $t\in[0,\tau]$, then the total work done by the particles on the fluid is equal to 
\[
\sum_{i=1}^n \int_0^\tau f_i(t)\cdot \dot X_i(t)\, dt = \eta  \sum_{i=1}^n \int_0^\tau \dot X_i(t)^2 \, dt
= \frac{\eta}\tau  \sum_{i=1}^n\int_0^1 \bigl(\partial_s X_i(\tau s)\bigr)^2\, ds.
\]
This work is converted into heat by friction with the molecules of the fluid. 

Writing $\rho_s = n^{-1} \sum_i \delta_{X_i(\tau s)}$ for the empirical measure of these particles, we have, using the formula for the Wasserstein tensor of atomic measures,
\[
\frac{\eta }\tau \sum_{i=1}^n \int_0^1 \bigl(\partial_s X_i(\tau s)\bigr)^2\, ds
\stackrel{\eqref{eq:Wass-norm-atomic}}= \frac{\eta n}\tau \int_0^1 \|\dot \rho_s\|_{-1,\rho_s}^2 \, ds
\stackrel{\eqref{def:WassersteinDistance}}\geq \frac{\eta n}\tau W_2(\rho_0,\rho_1)^2.
\]
This lower bound is achieved by a particle motion along straight lines with constant velocity.

\medskip

This observation shows us that 
\begin{example}
\noindent{The Wasserstein distance for empirical measures of particles can be interpreted as the  minimal energy dissipated by moving those particles through a viscous fluid. }
\end{example}
A similar statement holds locally in time: the Wasserstein tensor $\|\dot\rho\|_{-1,\rho}^2$, multiplied by the coefficient $n\eta$, is the power (work per unit of time) expended by the movement of~$n$ particles. The dimensional version of this statement is found by again connecting a concentration $c$ with a particle empirical measure $\rho$ by $\rho = c/\overline c
\,\Lebesgue|_\Omega$ and $n=\overline c N_A$ (see Section~\ref{sec:dimensional-versions-free-energy}). Then, writing  $\dot \rho = -\div \rho v = - \div (cv/\overline c)$, we have 
\begin{equation}
\label{eq:dissipation-Wasserstein-explanantion}
n\eta \|\dot \rho\|_{-1,\rho}^2 = n\eta \int \rho |v|^2 
= \frac {n\eta}{\overline c} \int {|w|^2}c \, dx 
= \eta N_A \int {|w|^2}c \, dx.
\end{equation}

\medskip

This observation provides a first motivation of the choice of the {Wasserstein metric tensor} in~\eqref{def:dissipation-potential-diffusion},
\[
\widetilde\Psi(c,w) = \frac{\tilde \eta}2 \int_\Omega {|w(x)|^2}{c(x)}\, dx,
\qquad \text{where }\dot c = -\div cw.
\]
By the argument above, $\widetilde\Psi(c,w)$ is one-half of the energy dissipation associated with moving the particles of concentration $c$  according to $w$, where $\tilde \eta = \eta N_A$ is a macroscopic friction coefficient. If one is willing to accept that the free energy $\aF$ in~\eqref{choice:entropy-diffusion-particles-fluid} is the driving force in a gradient flow, then modelling the dissipation with $\widetilde\Psi$ corresponds to assuming that dissipation of the free energy happens through viscous dissipation generated by the particles, as they move through a stagnant fluid.

\red{Wrong order: should do dissipation first, and Wasserstein distance second}

\medskip

While this argument has the advantage of being simple, it neglects the fact---put forward in the previous chapter---that the free energy, and especially the entropy terms, can only be understood from a stochastic point of view. In the next section we therefore develop a stochastic  view on the Wasserstein distance, by connecting it to the fluctuations of Brownian particles.

\section{Brownian particles and Wasserstein dissipation}

The previous section gave an interpretation of the Wasserstein distance in terms of viscous dissipation---conversion of work into heat by particles moving through a fluid.
We now focus on the role of the Wasserstein distance in \emph{stochastic} particle systems, and we  first consider a simple case.

Consider the system of $n$ particles $X_{\e,i}$ in $\R^d$ ($i=1,\dots,n$) defined as a rescaled Brownian motion: 
\begin{equation}
\label{def:Brownianparticles}
dX_{\e,i}(t) =\sqrt{2}\,\e\,dW_i(t), \qquad X_{\e,i}|_{\{t=0\}} = x_{0,i}.
\end{equation}
where for each $i$, $W_i$ is a Brownian motion in $\R^d$, $\sigma>0$ is a mobility coefficient, and the vector of positions $x_0 = (x_{0,1},\dots,x_{0,n})\in \R^{n d}$ is fixed.
If we fix $\tau>0$, then  by Schilder's theorem (e.g.~\cite[Th.~5.2.3]{DemboZeitouni98}), the process $\{X_{\e}: t\in [0,\tau]\}$,  satisfies a large-deviation principle for fixed $n$ and small $\e$, 
\[
\Prob\bigl(X_\e|_{[0,\tau]} \approx x|_{[0,\tau]} \bigr) \sim \exp \Bigl[-\frac1{\e^2} I(x)\Bigr],\qquad \text{as }\e\to0,
\]
with {rate functional}
\begin{equation*}
I(x)=\frac{1}{4}\int_0^\tau\abs{\dot{x}(t)}^2\,dt
\qquad\text{provided }x|_{t=0} = x_0.
\end{equation*}

Now note that the exponent can be rewritten, again using the notation $\rho_n = \frac1n \sum_i \delta_{x_i}$, as 
\[
\frac{1}{4\e^2}\int_0^\tau\abs{\dot{x}(t)}^2\,dt
= \frac{1}{4\e^2}\sum_{i=1}^n\int_0^\tau\abs{\dot{x_i}(t)}^2\,dt
= \frac{n}{4\e^2}\int_0^\tau \|\dot \rho_n(t)\|^2_{-1,\rho_n(t)}\,dt.
\]
By the same argument as in the previous section, this expression has a sharp lower bound 
\begin{equation}
\label{ineq:Brownianpart-I-W}
\frac1{\e^2} I(x) \geq 
\frac{n}{4\e^2\tau } W_2\bigl(\rho_n(0),\rho_n(\tau)\bigr)^2.
\end{equation}
Note how the this expression equals the minimal dissipated energy, when $\eta$ is equal to $1/4\e^2$. 

L\'eonard~\cite{Leonard12} was the first to prove the corresponding statement for the many-particle limit of empirical measures of (constant-mobility) Brownian particles, as follows. Let $X_{1,i}$ be the set of Brownian particles defined in~\eqref{def:Brownianparticles} with $\e=1$, and again let $\rho_n = \frac1n \sum_i \delta_{X_{1,i}}$ be the empirical measure. Then
\begin{equation}
\label{ldp:Brownian-abstract}
\Prob\big(\rho_n(t) \approx \rho^1 | \rho_n(0) \approx \rho^0\big) 
\sim \exp\bigl[-nI_t(\rho^1|\rho^0)\bigr] \qquad\text{as } n\to\infty,
\end{equation}
and the rate function $I_t(\,\cdot\,|\,\cdot\,)$ satisfies
\[
tI_t(\,\cdot\, | \rho^0) \to \frac14 W_2(\rho^0,\rho^1)^2 \qquad\text{as } t \downarrow 0
\]
in the sense of Gamma-convergence (see also~\cite{AdamsDirrPeletierZimmer11,DuongLaschosRenger13,PeletierRenger11TR}).
Written informally, this result states that 
\begin{equation}
\label{ldp:W2}
\Prob\big(\rho_n(t) \approx \rho^1 | \rho_n(0) \approx \rho^0\big) 
\sim \exp\Bigl[-\frac n{4t} W_2(\rho^0,\rho^1)^2\Bigr]
\qquad\text{as } n\to\infty \text{ and then } t\downarrow 0.
\end{equation}
This is the sense in which 
\begin{quote}
\emph{The Wasserstein distance characterizes the stochastic mobility of empirical measures of systems of Brownian particles. }
\end{quote}

\section{Brownian particles and Wasserstein dissipation, take 2}
\label{sec:Brownian-Wasserstein-2}

The large-deviation characterization~\eqref{ldp:W2} shows the connection between the Wasserstein distance and Brownian particles, but it doesn't yet explain how the whole gradient-flow arises. For this, it turns out, we need to include `the next term' in the small-$t$ asymptotic expansion of $I_t$. In~\cite{AdamsDirrPeletierZimmer11} we first proved that
\begin{equation}
\label{eq:ADPZ}
I_t(\,\cdot\, | \rho^0) = \frac1{4t} W_2(\rho^0,\rho^1)^2 + \frac12 \Ent(\,\cdot\,) - \frac12 \Ent(\rho^0) + o(1) \qquad\text{as } t \downarrow 0,
\end{equation}
in the sense of Gamma-asymptotic developments~\cite{AnzellottiBaldo93}. This was later generalized to a larger class of systems in~\cite{DuongLaschosRenger13}.

The function on the right-hand side of~\eqref{eq:ADPZ} is well known in the theory of gradient flows, as the basis for a discrete-time approximation of a gradient flow. Construct a sequence $\rho^k$ in an iterative fashion: fix $h>0$, and for each $\rho^{k-1}$  let $\rho^k$ be the solution of 
\[
\min_\rho \frac1{2h} W_2(\rho,\rho^{k-1})^2 + \Ent(\rho), 
\]
starting from some $\rho^0$. Then the piecewise-constant interpolation of this sequence with time step $h$, i.e.\ the function $t\mapsto \rho^{\lfloor t/h\rfloor}$, is an approximation of the Wasserstein gradient flow of~$\Ent$~\cite{JordanKinderlehrerOtto98}.

\section{Interpretation}

Equations~\eqref{ldp:Brownian-abstract} and~\eqref{eq:ADPZ} together illustrate how the dissipation metric $W_2$ and the entropy $\Ent$ together create a macroscopic gradient-flow behaviour. We have two competing phenomena:
\begin{enumerate}
\item For short time, Brownian particles prefer to stay put: this is represented by the large-deviation contribution
\[
\exp \Bigl[-\frac{n}{4t}W_2(\rho^1,\rho_0)^2\Bigr].
\]
\item Macrostates with lower entropy $\Ent$ contain more microstates; this leads to the large-deviation contribution
\[
\exp\Bigl[-\frac n2 \bigl(\Ent(\rho^1)-\Ent(\rho^0)\bigr)\Bigr]
\]
\end{enumerate}
The two effects independently behave like large-deviation effects; because both exponents scale with $n$, they are able to compensate each other, and the resulting behaviour is a mixture of the movement-aversion characterized by $W_2$ and the movement-preference characterized by $\Ent$. This mixture is the gradient flow.

\section{Revisit the derivation of solute diffusion}
\label{sec:solute-diffusion-revisited}

Let us now walk through the derivation of the convection-diffusion equation in Section~\ref{sec:diffusion-particles-fluid}, and interpret the choices made there in the light of the last two chapters. We repeat the relevant parts. 

\medskip

\nb{State space: } We choose a bounded set $\Omega\subset \R^d$ to be the container with the fluid; the particles are represented by their concentration $c:\Omega \to [0,\infty)$ (in moles per $\un m^3$). The state space is therefore $\aZ:= \{c\in L^1(\Omega): c\geq0\}$. Positions in $\Omega$ are labeled $x$.

\nb{Energy: } We choose as driving functional for this system the \emph{free energy} 
\begin{equation}
\label{re-choice:entropy-diffusion-particles-fluid}
\aF: \aZ\to\R, \qquad \aF (c) := RT \int_\Omega c(x)\log \frac{c(x)}{c_0}\, dx + \rho g \int_\Omega x_3 c(x)\, dx,
\end{equation}
where $c_0>0$ is an arbitrary reference concentration, and $\rho$ is the mass density contrast with the fluid.

\nb{{Processes:} }We allow the state $c\in \aZ$ to change through the effects of a velocity $w:\Omega \to\R^d$  by
\begin{equation}
\label{re-modass:c-w}
\dot c + \div cw =0  \quad \text{in $\Omega$},\qquad w\cdot n = 0 \quad \text{on }\partial\Omega,
\end{equation}
or in weak form,
\[
\forall \varphi \in C^1_b(\Omega): \qquad
\partial_t \int_{\Omega} \varphi(x) c(t,x)\, dx - \int_\Omega c(t,x)w(t,x)\nabla \varphi(x)\, dx =0.
\]

\nb{Dissipation potential: }We define the \emph{dissipation potential} on the set of process vectors $w$ as the functional 
\begin{equation}
\label{re-def:dissipation-potential-diffusion}
\widetilde\Psi(c,w) := \frac\eta2 \int_\Omega {|w(x)|^2}{c(x)}\, dx.
\end{equation}
\red{Doesn't explain what to take for the coefficient $\eta$.}

\medskip

We now understand that these choices mean the following. 
\begin{itemize}
\item The free energy~\eqref{re-choice:entropy-diffusion-particles-fluid} arises because we consider \emph{concentrations}; there is no energy term because (a) the particles are independent, and (b) there is no distinction between the points of $\Omega$, from the point of view of the particles. 
\item The dissipation potential is {(one-half of)} the heat generated by a particle movement~$w$ (see page~\pageref{dissipation-dissipation-potential} for the `one-half').
\item The combination of the two can be motivated by the fact that both {are large-deviation rate functionals at the same speed $n$}.
\end{itemize}



\section{Large deviations for SDEs}
\label{sec:SDE}

We can take the explanation of Section~\ref{sec:Brownian-Wasserstein-2} a step further by considering large deviations of the time course of empirical measures.
At the same time we extend the previous example by including interaction of the particles with a background potential $\Vb$ and with each other via an interaction potential~$\Vi$. Specifically, 
we take a system of $n$ particles described by 
\begin{equation}
\label{eq:SDE}
d X_i(t) = -\nabla \Vb(X_i(t))\,dt - \frac1n \sum_{j=1}^n \nabla \Vi(X_i(t)-X_j(t))\, dt + \sqrt 2\, \, dW_i(t).
\end{equation}
The continuum limit, as $n\to\infty$, of this system is the equation
\begin{equation}
\label{eq:interaction}
\dot\rho =  \Delta \rho 
+ \div \rho  \nabla \bigl[\Vb + \rho\ast \Vi\bigr].
\end{equation}
The large-deviation rate functional describing fluctuations of the system is  given by (see~\cite[Theorem~13.37]{FengKurtz06} and~\cite{DawsonGartner87})
\begin{equation}
\label{I:FK}
I(\rho) := \frac14 \int_0^T \Bigl\|\dot \rho - \Delta \rho 
- \div \rho \nabla \bigl[\Vb + \rho\ast \Vi\bigr]\Bigr\|^2_{-1,\rho}\, dt.
\end{equation}

We now connect this expression to the gradient-flow structures defined in Section~\ref{sec:GF-formal}. 
\begin{enumerate}
\item First, define
the free energy $\aF$ as the sum of entropy and potential energy for this system, as in Sections~\ref{subsec:free-energy} and~\ref{sec:Free-energy-generalizations}:
\begin{equation}
\label{def:free-energy-interacting-particles}
\aF(\rho) := \Ent(\rho) + \int_{\R^d}\left[\rho\Vb + \frac12 \rho(\rho\ast \Vi)\right].
\end{equation}
For this functional $\aF$, the Wasserstein gradient $\aK\D \aF$ is equal to (minus) the right-hand side of~\eqref{eq:interaction}:
\[
\aK(\rho)\D\aF(\rho) = \aK(\rho)\bigl[\log \rho + 1 + \Vb + \rho*\Vi\bigr]
= -\div \rho\nabla\bigl[\log \rho + 1 + \Vb + \rho*\Vi\bigr] .
\]
\item Second, expanding the square in~\eqref{I:FK} and using the chain rule~\eqref{eq:Wasserstein-chain-rule}, $I$ can therefore be written as 
\begin{equation*}
\label{eq:FK-FDT}
2I(\rho) 
= \aF(\rho(T))- \aF(\rho(0)) + \frac12 \int_0^T 
\left[ \|\dot \rho\|^2_{-1,\rho} 
 + \|-\D\aF(\rho)\|^2_{1,\rho}
\right] \, dt,
\end{equation*}
Note that $\|\dot \rho\|_{-1,\rho}$ is the Wasserstein dissipation, and $\|\cdot \|_{1,\rho}$ is the dual Wasserstein norm.
Therefore the rate function $I$ is \emph{exactly} the `global definition' of the gradient flow of Definition~\ref{def:GF-global} (actually, up to a factor $2$). 
\end{enumerate}

Finally, note that if we scale time by setting $t = Ts$, for $s\in [0,1]$, then in terms of the rescaled time $s$ we can write
\[
I(\rho) 
= \frac1{4T} \int_0^1 
 \|\partial_s \rho\|^2_{-1,\rho}\, ds + \frac12\bigl[\aF(\rho(T))- \aF(\rho(0))\bigr] +  
 \frac T4\int_0^1 \|\mathord-\D\aF(\rho)\|^2_{1,\rho}, ds.
\]
Compare this with the expression~\eqref{eq:ADPZ}: the term $W_2(\rho^1,\rho^0)^2$ is the infimum over all curves connecting $\rho^0$ to $\rho^1$ in time $1$, and therefore corresponds to the first term above; for the example of purely Brownian particles of Section~\eqref{sec:Brownian-Wasserstein-2}, $\aF = \Ent$, and therefore the second terms of the two expressions are identical; and a separate argument shows that the third term above indeed is expected to be small as $T\to0$~\cite{DuongLaschosRenger13}.

\section{Geometry and reversibility}

\def\diffc{A}
There are interesting connections between the geometry of the Brownian noise, the reversibility of the stochastic process, and the question whether the resulting evolution equation is a gradient flow or not.

This becomes apparent when we modify the system of the previous section by introducing a diffusion matrix $\diffc\in \R^{d\times d}$ and a mobility matrix $\sigma\in \R^{d\times d}$, as follows:
\begin{equation}
\label{eq:SDE-disc}
d X_i(t) = -\diffc\nabla \Vb(X_i(t))\,dt - \frac1n \sum_{j=1}^n \diffc\nabla\Vi(X_i(t)-X_j(t))\, dt + \sqrt 2\, \sigma\, dW_i(t).
\end{equation}

The large-deviation rate functional  of the system is similarly given by 
\begin{equation}
\label{I:FK-disc}
I(\rho) := \frac14 \int_0^T \Bigl\|\dot\rho - \div \sigma\sigma^T\nabla \rho 
- \div \rho \diffc\nabla \bigl[\Vb + \rho\ast \Vi\bigr]\Bigr\|^2_{-1,D(\rho)}\, dt,
\end{equation}
where the norm $\|\cdot\|_{-1,D(\rho)}$ is induced by the inner product
\[
(s_1,s_2)_{-1,D(\rho)} := \int w_1\cdot w_2\, D(\rho)\, dx,
\]
with $D(\rho) = \rho\sigma\sigma^T$. 
As before, the hydrodynamic limit of this system is the minimiser of~$I$,
\begin{equation}
\label{eq:interaction-disc}
\dot\rho =  \div \sigma\sigma^T\nabla \rho 
+ \div \rho \diffc \nabla \bigl[\Vb + \rho\ast \Vi\bigr].
\end{equation}

With the additional parameter freedom in $\diffc$ and $\sigma$, it is not always possible to write~\eqref{I:FK-disc} in the form of~\eqref{ineq:def-GGF-integral}. This depends on whether the cross term in~\eqref{I:FK-disc} is an exact differential, i.e., whether there exists a functional $\E$ such that
\begin{equation*}
\Bigl(\dot\rho,-\div \sigma\sigma^T\nabla \rho 
- \div \rho \diffc \nabla \bigl[\Vb + \rho\ast \Vi\bigr]\Bigr)_{-1,D(\rho)} = \partial_t \E(\rho).
\end{equation*}
This is the case if and only if $\sigma\sigma^T$ is a positive multiple of $\diffc$, a condition that is familiar from the fluctuation-dissipation theorem, also known as the Einstein relation. In that case, and writing $\sigma\sigma^T= kT\diffc$ for some `temperature' $T>0$ and the Boltzmann constant $k$, 
\begin{equation*}
-\div \sigma\sigma^T\nabla \rho 
- \div \rho \diffc \nabla \bigl[\Vb + \rho\ast \Vi\bigr]
= \aK(\rho) \D\aF(\rho),
\end{equation*}
where $\aK(\rho)\xi$ now is defined as $-\div D(\rho)\nabla\xi$ and the free energy $\aF$ is a modification of~\eqref{def:free-energy-interacting-particles},
\begin{equation*}
\aF(\rho) := \Ent(\rho) + \frac1{kT}\int_{\R^d}\left[\rho\Vb + \frac12 \rho(\rho\ast \Vi)\right].
\end{equation*}
Then the rate functional $I$ can be written in the form~\eqref{ineq:def-GGF-integral} as 
\begin{equation*}
\label{eq:FK-FDT-disc}
2I(\rho) 
= \aF(\rho(T))- \aF(\rho(0)) + \frac12 \int_0^T 
\left[ \|\dot \rho\|^2_{-1,D(\rho)} 
 + \|\mathord-\D\aF(\rho)\|^2_{1,D(\rho)}
\right] \, dt
\end{equation*}
and the evolution equation~\eqref{eq:interaction-disc} is the (modified, $D$-) Wasserstein gradient flow of~$\aF$.

\medskip
Our freedom to choose $\diffc$ and $\sigma$ separately gives us the insight that for this system the following four statements are equivalent:
\begin{enumerate}
\item $\sigma\sigma^T= kT\diffc$ for some $T>0$;
\item The evolution~\eqref{eq:interaction-disc} is a $D(\rho)$-Wasserstein gradient flow of $\aF$;
\item The rate functional $I$ can be written in the form~\eqref{ineq:def-GGF-integral};
\item For any  finite number $n$ of particles, the system~\eqref{eq:SDE-disc} is reversible.
\end{enumerate}
This equivalence, which holds for this specific system, suggests a much deeper connection between  reversibility and gradient-flow structure, that we comment on in detail in~\cite{MielkePeletierRenger13TR}.

\section{Comments}
\label{sec:comments-on-Wasserstein-dissipation}

\paragraph{Validity of  Stokes' law.}
A natural criticism of the law~\eqref{eq:linear-drag-law} with $\eta=6\pi\nu r$ would be that it is derived for a macroscopic spherical particle in a continuum viscous fluid, with no-slip boundary conditions. Real particles are not spherical and they are embedded in a sea of other molecules, which may even be of similar size. This is a valid point. 

Surprisingly,  Stokes' law is fairly robust under such generalizations. Molecular-dynamics simulations of hard-sphere particles in a hard-sphere `fluid' show the same law, regardless of the size of the particles (see e.g.~\cite[Fig.~5]{BocquetHansenPiasecki94}). In general, the law holds, with the same coefficient, for particles no smaller than a few times the size of the surrounding particles; below this size the coefficient may be different. See e.g.~\cite{Li09} and the references therein.

\red{
\paragraph{Fick's law.}
As we remarked in Section~\ref{sec:diffusion-particles-fluid}, the beautifully simple Fick's law seems to have vanished from the modelling of diffusion. 

Check the Tsallis entropies and such for reasoning behind linear Fick's law}

%
%

\chapter{Further Examples}
\label{ch:further-examples}

\section{The Allen-Cahn and Cahn-Hilliard models}
The Allen-Cahn or Cahn-Hilliard energy is the functional
\[
\aF: H^1(\Omega) \to [0,\infty], \qquad
\aF(u) = \frac12 \int_\Omega |\nabla u|^2 +\int_\Omega W(u),
\]
where $W$ is a double-well potential on $\R$ with wells at $\pm 1$ of depth $0$; the canonical example is $W(s) = \frac14 (1-s^2)^2$. Two important gradient flows constructed from $\aF$ are the $L^2(\Omega)$- and $H^{-1}(\Omega)$-gradient flows:
\begin{alignat*}3
&L^2(\Omega)\text{-gradient flow: } &\partial_t u &= \Delta u - W'(u) &\qquad &\text{(Allen-Cahn equation)}\\
&H^{-1}(\Omega)\text{-gradient flow: } &\qquad\partial_t u &= -\Delta(\Delta u - W'(u)) &\qquad &\text{(Cahn-Hilliard equation)}
\end{alignat*}
These correspond to taking $\aZ = H^1(\Omega)$ and $\Psi(\dot u) =  \frac12 \|\dot u\|_{L^2(\Omega)}^2$ or $\Psi(\dot u) = \frac12 \|\dot u\|_{H^{-1}(\Omega)}^2$. 

There is much to be said about the modelling of these systems, since there are various different modelling routes that all lead to these same two equations. In a future version of these notes we will return to this. \red{Do this} In the meantime, one modelling route makes use of \emph{systems of multiple components with volume exclusion}; we discuss this situation first in the next section.

\section{Multi-component diffusion with volume constraint}
Consider $m$ species $X_1,\dots,X_m$ with molar concentrations $c_1,\dots c_m$; assume that each species has a molar volume  $\alpha_i$ (in $\un {m^{\mathit d}/mol}$), so that $\alpha_ic_i$ is a volume fraction with dimension~$1$; for a given set $A\subset \R^n$, $\int_A \alpha_i c_i$ is the volume (in $\un {m^{\mathit d}}$) of the subset of $A$ occupied by species $X_i$.  We assume that the species diffuse in a domain $\Omega\subset \R^d$, with diffusion rates that depend on the species, but while preserving the total local volume. We also assume that the species do not enter or leave $\Omega$.

The volume constraint will take the form that $\sum_{i=1}^m \alpha_ic_i(x)$ should be constant and equal to $1$ everywhere in $\Omega$; the is the requirement that the complete mixture fills the whole space.

We now go through the same modelling steps as before. 

\nb{State space: } The state space is the set of $m$-tuples $\bc = (c_1,\dots,c_m)\in L^1_{\geq0}(\Omega)^m$ denoting molar concentrations.

\nb{Energy: } Since only entropy drives the diffusion, the natural choice for the driving functional is 
\[
\aF(\bc) = RT \sum_{i=1}^m \int_\Omega c_i(x) \log \frac{c_i(x)}{c_0}\, dx,
\]
where, as in Section~\ref{sec:dimensional-versions-free-energy}, $R$ is the universal gas constant, $T$ the temperature, and $c_0$ an arbitrary reference concentration.

\nb{Process space: } We have seen in Section~\ref{sec:solute-diffusion-revisited} that the natural process space  for a concentration is a space of $w_i:\Omega\to\R^d$, such that 
\[
\dot c_i + \div c_i w_i = 0, \qquad w_i \cdot n = 0 \quad \text{on }\partial\Omega.
\]
We write this slightly differently, in terms of fluxes $j_i = c_iw_i$:
\begin{equation}
\label{ass:ex:OttoE:cj}
\dot c_i + \div j_i = 0, \qquad j_i \cdot n = 0 \quad \text{on }\partial\Omega.
\end{equation}
However, since we want to enforce a volume constraint, we require that 
\begin{equation}
\label{ex:OttoE:prop}
\partial_t \sum_{i=1}^m \alpha_i c_i = -\sum_{i=1}^m \alpha_i \div j_i \stackrel{(*)}= 0.
\end{equation}
There are (at least) two ways of enforcing the equality $(*)$:
\begin{enumerate}
\item By enforcing $(*)$ itself: $\sum_{i=1}^m \alpha_i \div j_i = 0$ \emph{(global balance)};
\item By enforcing the stronger condition $\sum_{i=1}^m \alpha_i j_i = 0$ \emph{(local balance)}.
\end{enumerate}
The difference between the two properties can be recognized as follows: the first allows the volume constraint to be satisfied through some large-scale rearrangement of the species, while the second enforces that the total local volume flux is conserved. Otto and E discuss the two properties and their consequences at length in~\cite{OttoE97}.

\nb{Dissipation potential: } We have seen in Section~\ref{sec:solute-diffusion-revisited} that a natural dissipation potential for the diffusion of solutes is (after generalization to multiple species)
\[
\widetilde \Psi(\bc;\bj) = \sum_{i=1}^m \frac{\eta_i}2 \int_\Omega \frac1{c_i}{|j_i|^2},
\]
where $\eta_i$ are parameters specifying the relative ease of diffusion.

\begin{remark}
There is something fishy about the modelling of this dissipation potential. Recall that the derivation of the dissipation in Chapter~\ref{ch:Wasserstein-dissipation} was based on the Stokes law that characterizes how a single spherical particle moves through an otherwise stationary fluid. In the current case, however, the remaining fluid can not be considered stationary. To give an example, consider the `local-balance' condition in the case of two fluids of equal-size particles that each fill half of the space. From the point of view of a single particle of fluid~A, half of the remaining particles (the A particles) are performing macroscopically the same movement, and therefore provide no friction; the other half (the B particles) are moving in the opposite direction, with the same velocity (due to the local-balance condition) and therefore the relative velocity is double the velocity of the particle! A faithful modelling of the dissipation in this case should be based on a  more detailed description of exactly how the different particles organize themselves at a small scale. This can be a tough problem, because of subtle attraction and repulsion effects, and resuling small-scale patterning.

For the moment, however, we stick with this expression, since we prefer to focus on the consequences of the choice for local vs.\ global balance.
\end{remark}

\nb{Derive the equations: } For the case of \textbf{global} balance, the equations are given by the minimization problem
\[
\min_{\bj,\dot \bc}\  \Bigl\{ \widetilde \Psi(\bc;\bj) + \langle \aF'(c), \dot c\rangle:\ \dot \bc \text{ and } \bw \text{ connected by }\eqref{ass:ex:OttoE:cj}\text{ and $\bj$ satisfies global balance}\Bigr\},
\]
with stationarity condition
\[
0 = \sum_{i=1}^m \int_\Omega \biggl\{\eta_i \frac{j_i}{c_i} \tilde \jmath_i - RT \log {c_i}{c_0} \div \tilde \jmath_i + p\alpha_i \div \tilde \jmath_i\bigg\}
\]
for all $\tilde \jmath_i$, where $p$ is a Lagrange multiplier. It follows that 
\[
j_i = \frac{1}{\eta_i} \Bigl(- RT\nabla c_i + \alpha_i c_i \nabla p\Bigr),
\]
with resulting evolution equations 
\[
\dot c_i = \frac1{\eta_i} \Bigl(RT\Delta c_i -\alpha_i \div c_i\nabla p\Bigr).
\]
By~\eqref{ex:OttoE:prop} and~\eqref{ass:ex:OttoE:cj} the Lagrange multiplier $p$, which has the interpretation of a pressure, satisfies
\[
\div\Bigl(\sum_{i=1}^m \frac{\alpha_i^2 c_i}{\eta_i} \nabla p\Bigr) = RT \sum_{i=1}^m \frac{\alpha_i}{\eta_i} \Delta c_i,
\]
with boundary condition
\[
\partial_n p  = \Bigl(\sum_{i=1}^m \alpha_i c_i\Bigr)^{-1} RT \sum_{i=1}^m \partial_n c_i.
\]

\medskip
For \textbf{local} balance, a similar reasoning leads to a vector-valued Lagrange multiplier $\lambda: \Omega\to\R^d$, and the equations
\[
\dot c_i = \frac1{\eta_i} \Bigl(RT\Delta c_i -\alpha_i \div c_i\lambda\Bigr),
\]
and 
\[
\lambda = -RT \sum_{i=1}^m \frac1{\alpha_ic_i} \nabla c_i.
\]

\subsection{Discussion}

\paragraph{Local and global balance}
How to choose between the local and global balance condition? Since the local condition is more stringent than the global one, an Occam-razor-type argument suggests that one should choose the global condition unless there is a good reason to choose the local one. 

One good reason for choosing the local condition is when the microscopic dynamics is incapable of producing large-scale motion. An example of this is Kawasaki exchange dynamics, in which particles on a lattice exchange with their neighbours when a stochastic clock rings. Since each individual exchange is purely local, there is no mechanism to create large-scale exchanges; and after upscaling such a dynamics will lead to a local-balance condition.

\section{A moving vesicle in a viscous fluid with diffusing solutes} 
\label{sec:ex:vesicle}


In this example $\Omega\subset \R^d$ is not fixed, but moves around in some viscous fluid. A typical case is a vesicle (a biological container bounded by a membrane of lipids) immersed in water. Inside $\Omega$ some chemical is dissolved in the fluid, which can diffuse freely inside $\Omega$, but can not pass the boundary $\partial\Omega$. The surrounding fluid can pass through $\partial \Omega$ albeit with some resistance. 

We now go through the same modelling steps as before.

\nb{State space: } The state is a set of pairs:
\[
\aZ := \{(\Omega,c): \Omega\subset \R^d, \ c\in L^1(\R^d), \ \supp c\subset \Omega\}.
\]
The requirement that the solute stays within $\Omega$ is encoded in the definition of $\aZ$. 

\nb{Energy: } The boundary of the vesicle is assumed to have surface tension, which implies that the domain $\Omega$ has \emph{surface energy}. The driving force should therefore be the sum of (negative) entropy and surface energy, the \emph{free energy}
\[
\aF: \aZ\to\R, \qquad \aF (\Omega,c) := RT \int_\Omega c(x)\log \frac{c(x)}{c_0}\, dx 
+ \alpha |\partial \Omega|,
\]
where we write $|\partial\Omega|$ for the perimeter of $\Omega$ (the surface area in three dimensions, and the length of the boundary in two dimensions). The parameter $\alpha>0$ is the energy per unit area, which is often called the surface tension.

\nb{Processes: }We allow the state $(\Omega,c)\in \aZ$ to change through the effects of three objects, $w$, $u$, and $v_n$:
\begin{enumerate}
\item The evolution of $\Omega$ is characterized by the normal velocity $v_n$ of the boundary $\partial\Omega$ (the normal $n$ points outwards);
\item The water moves with velocity $u:\R^d\to\R^d$, which is required to be divergence-free, $\div u =0$;
\item The solute concentration evolves (as above) through a solute velocity $w$ with $\supp w \subset \Omega$:
\begin{equation}
\label{modass:Intro-osmosis}
\dot c + \div cw =0  \quad \text{in $\cD'(\R^d)$},
\end{equation}
or again in weak form,
\begin{equation}
\label{modass:Intro-weakform-osmosis}
\forall \varphi \in C^1_b(\R^d): \qquad
\partial_t \int_{\R^d} \varphi(x) c(t,x)\, dx - \int_{\R^d} c(t,x) w(t,x)\nabla \varphi(x)\, dx =0.
\end{equation}
Note that we now consider $w$ defined on the whole of $\R^d$, and the no-flux boundary condition of the previous case has vanished. Indeed the normal flux need not be zero on the boundary: if the boundary moves, then it will typically happen that nearby solutes move with it. A formal calculation with the weak form~\eqref{modass:Intro-weakform-osmosis} and the condition $\supp c, \supp w \subset \Omega$ shows that it indeed implies a weak version of the boundary condition
\begin{equation}
\label{modass:Intro-osmosis-bc}
w\cdot n = v_n \qquad\text{on }\partial\Omega.
\end{equation}
\end{enumerate}
Again these assumptions contain a number of modelling choices, such as the fact that the water is assumed to be incompressible. Another choice is hidden in the fact that $v_n$ may be different from  $u\cdot n$, i.e. that the boundary does not necessarily move with the fluid---this implies that the fluid can move through the boundary.

In this case a process vector is thus a triple $(w,u,v_n)$.

\nb{Dissipation potential: }We define the {dissipation potential} on the set of process vectors $(w,u,v_n)$ as the functional 
\begin{equation}
\label{def:dissipation-potential-osmosis}
\widetilde\Psi(\Omega,c;w,u,v_n) := \frac{\eta_c}2 \int_\Omega c{|w-u|^2}
  + \frac{\eta_u}2 \int_{\R^d} |\e(u)|^2  
  + \frac{\eta_b}2 \int_{\partial\Omega} (u\cdot n-v_n)^2.
\end{equation}
Here $\eta_c$, $\eta_u$, and $\eta_b$ are three friction-type constants; $\e(u) = \frac12(\nabla u + \nabla u^T)$ is the symmetric part of the gradient $\nabla u$. Note the different domains of integration.

The first term is the same as in the previous example, with a twist: we first subtract $cu$ from $w$, before we penalize the result. This has the effect of measuring this dissipation in terms of the difference with the convective flux $cu$. The idea is that simply convecting the solute along with the flow $u$ of the fluid, i.e. with flux $cu$, should not lead to diffusive dissipation; differences with respect to $cu$ should. The second term is a measure of  dissipation in the viscous fluid, as a result of shear in the velocity field $u$. The inclusion of this term therefore reflects the choice that fluid movement is accompanied by dissipation of energy through viscous friction. The third term similarly measures dissipation due to friction, but  in the boundary, in terms of the relative velocity $u\cdot n - v_n$ between fluid and boundary. The inclusion of this term therefore implies that water can move through the boundary, but doing so requires dissipation of energy.

\red{Doesn't explain what to take for the coefficient $\eta_c$, or for the others.}

\nb{Derive the equations: } We now apply the algorithm to find the equations of motion, by minimizing
\begin{equation}
\label{alg:min-osmosis}
(w,u,v_n) \mapsto \widetilde\Psi(\Omega,c;w,u,v_n) + \langle \aF'(\Omega,c), (w,u,v_n)\rangle.
\end{equation}
The term in angle brackets requires some specification. If $t\mapsto (\Omega_t,c_t)$ is a curve in $\aZ$, with process vector $(w,u,v_n)$ at $t=\tau$, then by the Reynolds transport theorem
\begin{align*}
\partial_t  \int_{\Omega_t} c_t\log\frac{c_t}{c_0}  \Bigr|_{t=\tau}  
  &= \int_{\Omega_\tau} \Bigl(\log\frac {c_\tau}{c_0}+1\Bigr) \dot c_t\bigr|_{t=\tau} 
   + \int_{\partial\Omega_\tau} c_\tau\log \frac{c_\tau}{c_0} \;v_n.
\end{align*}
The first integral on the right-hand side can be rewritten using~\eqref{modass:Intro-osmosis} and~\eqref{modass:Intro-osmosis-bc} into
\begin{align*}
-\int_{\Omega_\tau} \Bigl(\log \frac{c_\tau}{c_0}+1\Bigr) \div w
 &= \int_{\Omega_\tau} w\cdot {\nabla c_\tau}
   - \int_{\partial\Omega_\tau} \Bigl(\log \frac{c_\tau}{c_0}+1\Bigr) c_\tau w\cdot n\\
 &= \int_{\Omega_\tau} w\cdot{\nabla c_\tau}
   - \int_{\partial\Omega_\tau} \Bigl(\log \frac{c_\tau}{c_0}+1\Bigr) c_\tau v_n
\end{align*}
By the properties of the total curvature $H$ of the surface $\partial\Omega$ (i.e.\ the sum of principal curvatures, or equivalently $d$ times the mean curvature),
\[
\partial_t |\partial \Omega_t|\Bigr|_{t=\tau} = -\int_{\partial\Omega_\tau} H v_n.
\]
Therefore, and this is the definition of the angle-bracket term, 
\[
\langle \aF'(\Omega_\tau,c_\tau), (w,u,v_n)\rangle := 
\partial_t \aF(\Omega_t,c_t) \bigr|_{t=\tau} = 
RT \int_{\Omega_\tau} w\cdot {\nabla c_\tau} 
- \int_{\partial\Omega_\tau} \big[RT c_\tau + \alpha H\bigr]\,v_n.
\]

\medskip

Minimization of~\eqref{alg:min-osmosis} with respect of each of the three components leads to the system of equations
\begin{alignat*}2
\forall \tilde w: &\quad &0 &= \eta_c \int_\Omega c{(w-u)\cdot\tilde w}
  + RT\int_\Omega \tilde w\cdot {\nabla c} \\
\forall \tilde u: &\quad &0 &= \eta_c \int_\Omega c(u-w)\cdot \tilde u + \eta_u \int_{\R^d} \e(u):\e(\tilde u) + \eta_b \int_{\partial\Omega} (u\cdot n-v_n) \tilde u\cdot n - \int_{\R^d} p\div \tilde u\\
\forall \tilde v_n: &\quad &0 &= \eta_b \int_{\partial\Omega}(v_n-u\cdot n) \tilde v_n
- \int_{\partial\Omega} \big[RT c + \alpha H\bigr]\,\tilde v_n.
\end{alignat*}
Note the appearance of the pressure $p$ as a Lagrange multiplier associated with the constraint $\div u = 0$. Because of this Lagrange multiplier the second equation holds for each~$\tilde u$, even for those that are not divergence-free.

By performing similar manipulations on these equations as in the previous example, we derive that 
\begin{align*}
cw &= cu - \frac{RT}{\eta_c} \nabla c \quad\text{in }\Omega, \qquad  w\cdot n = v_n \text{ on } \partial\Omega,\\
v_n &= u\cdot n +\frac1{\eta_b} [RTc + \alpha H]\text{ on } \partial\Omega,
\end{align*}
and for $u$ we find the set of equations and boundary conditions
\begin{alignat*}2
-\div \sigma &= -RT\nabla c  & \qquad & \text{in }\Omega\\
-\div \sigma &= 0 & \qquad & \text{in }\R^d\setminus \Omega\\
\div u &= 0 &&\text{in }\R^d\\
[\sigma]\cdot n &= \eta_b(u\cdot n-v_n)n&& \text{on } \partial\Omega.
\end{alignat*}
Here $\sigma = \eta_u \e(u) - pI$, and $[\sigma]$ is the jump in $\sigma$ over $\partial\Omega$, i.e. $\sigma_{\mathrm{ext}}-\sigma_{\mathrm{int}}$.

Combining these expressions with~\eqref{modass:Intro-osmosis} and~\eqref{modass:Intro-osmosis-bc} we find
\begin{subequations}
\label{eq:osmosis}
\begin{alignat}2
-\div \sigma &= -RT\nabla c  & \qquad & \text{in }\Omega\label{eq:osmosis-u-inner}\\
-\div \sigma &= 0 & \qquad & \text{in }\R^d\setminus \Omega\label{eq:osmosis-u-outer}\\
\div u &= 0 &&\text{in }\R^d\\
\dot c &= -\div (cu) + \frac{RT}{\eta_c} \Delta c & \qquad & \text{in }\Omega\label{eq:osmosis-c}\\
[\sigma]\cdot n &= -RTc - \alpha H&& \text{on } \partial\Omega\label{eq:osmosis-sigmabc}\\
cv_n &= \Bigl(\frac{RT}{\eta_c} \nabla c - cu\Bigr)\cdot n   && \text{on } \partial\Omega\label{eq:osmosis-cbc}\\
v_n &= u\cdot n +\frac1{\eta_b} [RTc + \alpha H]&&\text{on } \partial\Omega.\label{eq:osmosis-vn}
\end{alignat}
\end{subequations}

\subsection{Discussion}

\paragraph{Interpretation.}
The equations~\eqref{eq:osmosis} combine a number of effects. The first three equations describe Stokes flow in the two domains~$\Omega$ and $\R^d\setminus \Omega$, driven by boundary forcing. Note that the term $RTc$ is the same as the \emph{osmotic pressure} given by the Van 't Hoff equation (see e.g.~\cite[Sec.~5.4]{AtkinsPaula06}). If necessary, it can be incorporated into the pressure $p$, since this pressure is a Lagrange multiplier and therefore not known explicitly. In that case the osmotic pressure disappears from equations~\eqref{eq:osmosis-u-inner} and~\eqref{eq:osmosis-sigmabc}; it still remains present in the boundary conditions~\eqref{eq:osmosis-cbc} and~\eqref{eq:osmosis-vn}. 

Equation~\eqref{eq:osmosis-c} is a traditional convection-diffusion very similar to the one we derived in Section~\ref{sec:diffusion-particles-fluid}; the only difference is the convective term $\div cu$, which arises as the result of the large-scale flow field $u$. 

Equation~\eqref{eq:osmosis-vn} gives a direct illustration of the phenomenon of osmotic swelling: on the right-hand side, apart from the passive convection term $u\cdot n$, there is the balance of the two pressures on the membrane. The osmotic pressure $RTc$ forces the membrane outwards, and for convex shapes the curvature term $-\alpha H$ acts in the opposite direction. 

\paragraph{Extensions.} In the case of a vesicle surrounded by a lipid bilayer, the assumed linear dependence of the energy on the surface area is rather crude. Actual lipid bilayers can extend only very slightly, and a better assumption would be a dependence of the type $\alpha'(|\partial\Omega|-A_0)^2$. With this choice, the derivation would similar with only $\alpha$ replaced by $2\alpha'(|\partial\Omega|-A_0)$. 

A further improvement could be to treat the lipid bilayer as a two-dimensional viscous incompressible fluid. In that case one would define the vectorial membrane velocity $v_{\mathrm{mem}}$, replace the normal velocity $v_n$ above by $v_{\mathrm{mem}}\cdot n$, and include a dissipation term of the form $\int_{\partial\Omega} |d|^2$, where $d$ is an appropriate measure of in-plane strain (see~\cite{ArroyoDeSimone09}).

\paragraph{Limiting cases.} It is instructive to consider setting constants to $0$ or $+\infty$, and compare the effect (a) on the free energy and the dissipation potential, and (b) on the equation~\eqref{eq:osmosis}. For instance, if $\eta_b=\infty$, which means in~\eqref{def:dissipation-potential-osmosis} that trans-membrane fluxes require infinite amounts of dissipated energy, then indeed  equation~\eqref{eq:osmosis-vn} reduces to $v_n=u\cdot n$ and there is no flux through the boundary. If $\eta_b=0$, on the other hand, implying that the fluid can move through the membrane without any friction, then equation~\eqref{eq:osmosis-vn} reduces to the force balance $RTc = -\alpha H$. This latter boundary condition is known as the~\emph{Gibbs-Thomson law}~\cite[Sec.~7.3.1.3]{BirdiS08}, and this discussion therefore gives some insight into the type of modelling choices that lead to this law.

Similarly, if $\alpha=0$, then the boundary can increase without limit, and indeed in~\eqref{eq:osmosis-sigmabc} and~\eqref{eq:osmosis-vn} the curvature then no longer appears. Or if $\eta_u=\infty$, then the fluid is infinitely viscous, and indeed equations~\eqref{eq:osmosis-u-inner}, \eqref{eq:osmosis-u-outer}, and~\eqref{eq:osmosis-sigmabc} give $u=0$.

\drop{
\section{Vesicles with stiff membranes}

%
%
%
%
%
%
%

\section{Cahn-Hilliard-type models}









\section{Chemical reactions}
\section{A chemical reaction}

The next example is a simple chemical reaction system,
\[
\begin{tikzcd}[column sep=small]
& A \arrow[leftharpoondown, shift right=0.4ex]{dl}[swap]{3}\arrow[rightharpoonup, shift left=0.4ex]{dl}
\arrow[leftharpoondown, shift right=0.4ex]{dr}\arrow[rightharpoonup, shift left=0.4ex]{dr}{1} & \\
C \arrow[leftharpoondown, shift right=0.4ex]{rr}[swap]{2}\arrow[rightharpoonup, shift left=0.4ex]{rr} & & B
\end{tikzcd}
\]

\nb{State space: }The system is characterized by the concentrations of the three species $A$, $B$, and $C$, so that the state space $\aZ$  is $\R^3$ with $z= (z_A,z_B,z_C)$.

\nb{Energy: } We choose $\aF (z) := \frac12 z^TAz$ where $A=\diag(a_A,a_B,a_C)$ is the matrix of \emph{activity coefficients}.

Why this is a good choice may not be obvious at this moment; we discuss this in more detail in Section~\ref{sec:chemistry}.

\medskip

Whereas in the previous example we characterized the dissipation directly in terms of the change of state, $\dot x$, in many cases it is useful to characterize the set of \emph{state changes} separately from the set of \emph{states}. This is called choosing the \emph{kinematics} of the system. 

\nb{Kinematics: } The three states change because of three chemical reactions, with rates (say) $r_1$, $r_2$, and $r_3$. The relationship between reaction rate and concentration change is given by the \emph{stoichiometric matrix} $N$, for which $N_{ij}$ is the change in concentration $i$ as the result of reaction $j$:
\[
\dot z = Nr \qquad\text{with}\qquad 
N =\begin{pmatrix} 
  -1 & 1 & 0\\
  0 & -1 & 1\\
  1 & 0 & -1
\end{pmatrix}.
\]
We call the vector $r =(r_1,r_2,r_3)^T$ the \emph{kinematic tangent vector} (see Section~\ref{sec:tangents}).

\nb{Dissipation mechanism: }Energy is dissipated in the reactions, and we now choose how. Each of the three reactions is parametrized by a global reaction rate parameter $k_i$, $i=1,2,3$, which may depend on all the concentrations $z$. The Onsager operator $\aK$ is now defined as \red{Have we defined the name Onsager operator yet?}
\begin{equation}
\label{def:aK-linearizedchemistry}
\aK(z) = N\begin{pmatrix}
  k_1(z) & 0 & 0\\
  0 & k_2(z) & 0\\
  0 & 0 & k_3(z) \end{pmatrix} N^T.
\end{equation}
The reasoning behind this choice can be understood as follows. We combine three elements of information:  (a) this operator characterizes the right-hand side in the equation $\dot z = -\aK(z)\aF'(z)$, (b) we chose $\dot z$ to be of the form $Nr$, and (c) $\aK$ is symmetric. Therefore it follows that $\aK$ should be of the form $\aK(z) = Nf(z) N^T$ for some matrix-valued function $f$. Here we make the choice that $f$ is diagonal. 

\nb{Result: } Putting all this together and applying~\eqref{eq:GF-intro-K} we find the equation
\begin{equation}
\label{eq:resultLinearizedChemstry}
\dot z = - N \begin{pmatrix} 
  k_1(z)(a_Az_A - a_Bz_B)\\
  k_2(z)(a_Bz_B - a_C z_C)\\
  k_3(z)(a_Cz_C - a_Az_A)
\end{pmatrix}
\end{equation}
In this formula we can see that taking $f$ diagonal implies that the global rate, which controls both forward and backward reactions, may be concentration dependent, but this dependence will be the same for the individual forward and backward rates.

\subsection*{Remarks} 

\paragraph{Modelling choices.}

\paragraph{Additivity of $\aK$.} Note how the final result~\eqref{eq:resultLinearizedChemstry} can be additively decomposed in three terms, corresponding to separate reactions, 
\[
\dot z = \aK_1(z)\aF'(z) +\aK_2(z)\aF'(z) + \aK_3(z)\aF'(z),
\]
where $\aK_j = NK_j N^T$ with $K_j$ a matrix with only $k_j$ as the $j$-th diagonal element. Note how the driving functional $\aF$ is the same for each term, but the Onsager operators are different. This is a general principle, that we discuss in more detail in Section~\ref{sec:additivity-Onsager}, and which will be useful in modelling more complex systems, since it will allow us to model various processes separately and add the results.

\paragraph{Conservation of mass.}
Note that the range of $N$ is orthogonal to $(1,1,1)^T$. This is another way of saying that the equation $\dot z = Nr$ conserves the total concentration $\sum_i z_i$, which is of course expected for this system. 

It also implies that $\aK$ is not invertible, and therefore $\aG$ is not well-defined as the inverse of $\aK$. However, the dissipation potentials $\Psi$ and $\Psi^*$ \emph{are} both well-defined. In this case $\Psi^*$ is naturally constructed from $\aK$,
\[
\Psi^*(z,\xi) = \frac12 \xi^T \aK(z) \xi,
\]
and $\Psi$ is determined by duality:
\[
\Psi(z,s) = \sup_\xi \Bigl[ \xi\cdot s - \frac12 \xi^T \aK(z) \xi\Bigr]
= \begin{cases}
\inf\bigl\{ \tfrac12 \sigma^T k^{-1} \sigma : s = N\sigma \bigr\}& \text{if $s\in \range(N)$}\\
+\infty & \text{otherwise}.
\end{cases}
\]
Here $k=\diag(k_1,k_2,k_3)$ is the diagonal matrix in~\eqref{def:aK-linearizedchemistry}. Note how the degeneracy of $\aK$---the lack of full rank---creates the singular nature of $\Psi$.

Of course there is an alternative way to implement the conservation of mass, by switching to the affine subspace of $\R^3$ consisting of vectors $z$ with prescribed $\sum_i z_i$. In this case the tangent space will be two-dimensional instead of three-dimensional, and $\aK$ and $\aG$ will both be non-singular. 

\section{Cross-diffusion}
\section{Diffusion of species on interfaces, fixed and moving}
\section{Phase-field models}
\section{Rate-independent systems}
\section{Quantum Mechanics}
\section{Other entropies}
\section{\red{Many others}}
\section{The $H^1$-$L^2$-gradient flow}

\red{Modify for new context}

Question~\ref{q:H1L2} of these notes was `where does the $H^1$-$L^2$ gradient-flow formulation of the diffusion equation come from?' (see page~\pageref{q:H1L2}). I don't have a convincing model for this, but actually the absense of a good model is even more interesting than its existence might have been.

In this section I describe my best shot. It's an unreasonable, unconvincing model, but it's the best way I have of connecting this gradient flow to a model.

In this model the energy and dissipation are mechanical, not probabilistic. The continuum system is the limit of discrete systems, consisting of $n$ spheres in a viscous fluid, arranged in a horizontal row and spaced at a distance from each other. Each sphere is constrained to move only vertically, so that the degrees of freedom are the $n$ vertical positions $u_i$. 

\begin{figure}[h]
\centering
\noindent
\psfig{figure=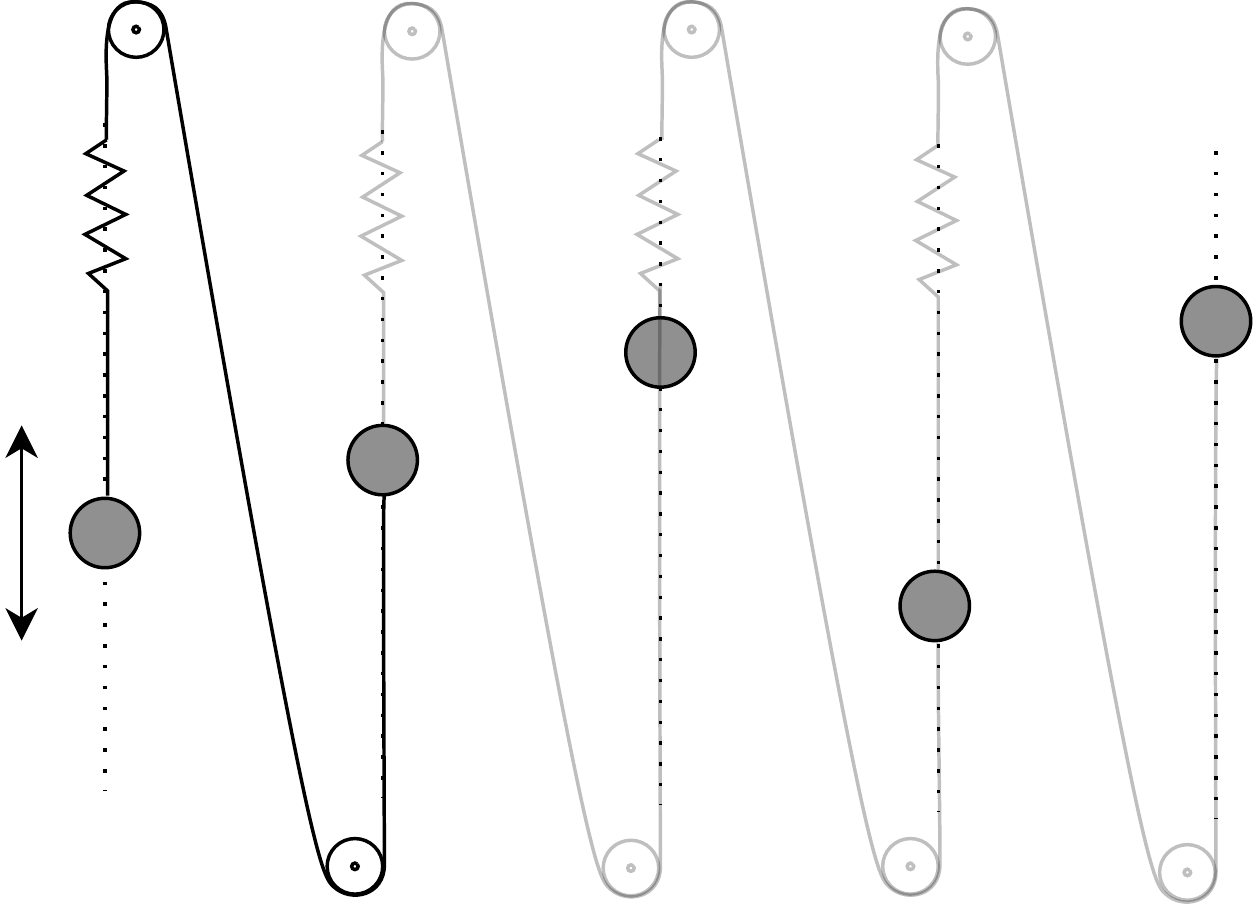,height=5cm}
\caption{The model of this section: spheres in a viscous fluid, only moving up and down, and connected by elastic springs. The spring measures the difference in displacement $u_i-u_{i+1}$.}
\end{figure}

The energy of this discrete system arises from springs that connect the spheres, in such a way that the energy of the spring connecting sphere $i$ to sphere $i+1$ is $k(u_i-u_{i+1})^2/2$. Therefore the total energy is
\[
\E_n(u_1,\dots, u_n) = \sum_{i=1}^{n-1} \frac k2 (u_i-u_{i+1})^2.
\]

Movement of the spheres requires displacement of the surrounding fluid, and therefore dissipates energy. The dissipation rate corresponding to a single sphere moving with velocity $v$ is assumed to be $cv^2/2$, where $c>0$ depends on the fluid. (Stokes' law gives $c=6\pi R\eta$, where $R$ is the radius of the sphere and $\eta$ the viscosity of the fluid). For a system of moving spheres the total dissipation rate is therefore
\[
\sum_{i=1}^n \frac c2 \dot u_i^2.
\]

We now construct a gradient flow by putting energy and dissipation together, resulting in the set of equations
\begin{align*}
c\dot u_1 &= k(u_2-u_1)\\
c\dot u_i &= k(u_{i+1}-2u_i + u_{i-1}) \qquad \text{for }i=2,\dots,n-1\\
c\dot u_n &= k(u_{n-1}-u_n)
\end{align*}

We can now take the continuum limit, either in the energy and dissipation, or in the equations above; both yield the same limit, which is the $L^2$-gradient flow of the energy 
\[
\E(u) := \int |\nabla u|^2, 
\]
and the diffusion equation~\eqref{eq:diffusion}. However, the important modelling choices are already clear in the discrete case. These two choices, viscous dissipation and interparticle coupling, are very reasonable in themselves, but the unnatural  restriction to vertical movement and the just-as-unnatural spring arrangement make this model a very strange one.

}

\chapter{Remarks and connections with thermodynamics}
\label{ch:thermodynamics}

\section{Boundary conditions}

The examples in these notes all focus on boundary conditions for which the system is \emph{closed}---in the sense that the boundary conditions do not add or take away energy. For other boundary conditions this will be different. 

Boundary conditions in models can play two roles. First there is the type of boundary conditions that reflects the interaction with the part of the world that is not represented in the model. For such boundary conditions it is very hard to give general guidelines, since practically anything is possible, and in general the driving functional need no longer be decreasing along solutions.

The second type of boundary condition interfaces with another part of the model. Examples of this are two domains connected by their boundaries, or free- and moving-boundary problems. In this case the natural approach is to include the boundary and its properties in the model, in the same way as the other aspects. As an example, in Section~\ref{sec:ex:vesicle} we consider a moving-boundary problem, in which the moving boundary is the wall of a vesicle. Fluid can move through the boundary, but in doing so it dissipates energy, and therefore the cross-boundary flow enters in the process vector and in the dissipation. As a result, boundary conditions at this moving interface arise in the derivation of the equations.

\section{Free energy from a thermodynamic point of view}

For this section~\cite[Sec.~5.1]{Schroeder07} is a good reference. 

\subsection{The `free' in free energy: \emph{available} work}

One thermodynamic situation in which `free energy' arises is that of zero pressure and constant temperature. The traditional thermodynamic definition of \emph{Gibbs free energy} in this context is $G = E-TS$, where $E$ is the energy of the system, $S$ the entropy (with the physically traditional sign) and $T$ the temperature. This is the same formula as we found in Section~\ref{sec:modelling-free-energy}.

Imagine that we create this system by starting with nothing, and slowly (quasistatically) adding material, performing work, and transferring heat, until we reach final energy $E$ and final entropy $S$. Since we do this under constant temperature, the traditional thermodynamic relation $dQ = TdS$ for quasistatic evolutions implies that we need to transfer exactly $TS$ heat from the surroundings into the system, if we normalize the empty system to have zero entropy. The remaining energy $E-TS$ then needs to be supplied in the form of work on the system, in such a way that the entropy does not increase.
Therefore $E-TS$ is that amount of work that we need to do to create the system at energy $E$ and entropy~$S$, taking into account that the environment can supply heat at temperature $T$.

The word `free' actually comes from performing the opposite operation, annihilating the system by performing work on the environment and transferring the entropy into the environment by heat transfer. `Free' then refers to the fact that the transfer of entropy requires transfer of heat, and therefore `costs' energy; the remaining energy $E-TS$ is `available' to perform work on the environment.

\subsection{Free energy as driving force}

Section~\ref{subsec:free-energy} explains the combination $E-TS$ in a different way: $E/T$ is an approximation of the change in the entropy of the heat bath (or more generally, the change in the entropy of the rest of the universe) as a result of withdrawing energy $E$ from the universe and putting it into the system. The sum $S - E/T$ then is the total energy of the system and the rest of the universe, i.e. of the whole universe.

This shows that entropy and free energy are really the same concept: the free energy is simply an approximation of the entropy of the universe taking into account only the behaviour of the system and the exchange of energy. 

The question why free energy should be a driving force for any given process is therefore the same as the question why entropy should be one. For this question I know of only one answer, which is the one given in Chapters~\ref{ch:entropy-free-energy-stationary} and~\ref{ch:Wasserstein-dissipation}. \red{Talk about Leli\`evre's argument}

\red{
\subsection{Dissipated free energy is transformed into heat}

Why is this? And why do we recover `maximal' work? Is that related to the gradient-flow optimization?
}

%
%
%
%
%


\appendix
\chapter{Elements of measure theory}
\label{app:measure-theory}


Throughout these notes we work with a measurable space $(\Omega,\Sigma)$ consisting of a topological set $\Omega$ with a $\sigma$-algebra $\Sigma$ that contains the  Borel $\sigma$-algebra $\B(\Omega)$. We always assume that $\Omega$ is complete, metrizable, and separable. $\M(\Omega)$ is the set of all (finite or infinite) signed measures on $(\Omega,\B(\Omega))$, and $\P(\Omega)$ is the set of all non-negative measures $\mu\in\M(\Omega)$ with $\mu(\Omega)=1$. For $\rho\in \P(\Omega)$, $L^2(\rho)$ is the set of measurable functions on $\Omega$ with finite $L^2(\rho)$-norm:
\[
\|f\|_{L^2(\rho)}^2 := \int_\Omega f^2 \, d\rho.
\]
A sequence of measures $\mu_n\in\M(\Omega)$ is said to converge \emph{narrowly} to $\mu\in\M(\Omega)$ if 
\[
\int_\Omega f\, d\mu_n \stackrel{n\to\infty}\longrightarrow \int_\Omega f\, d\mu
\qquad\text{for all }f\in C_b(\Omega).
\]
The Lebesgue measure on $\R^d$ is indicated by $\Lebesgue^d$.

The \emph{push-forward} of a measure $\mu$ by a Borel measurable mapping $\varphi:\Omega\to\Omega$ is the measure $\varphi_\#\mu$ defined by
\[
\varphi_\#\mu(A) := \mu(\varphi^{-1}(A))\qquad
\text{for any }A\in \B(\Omega).
\]
It satisfies the identity
\[
\int_\Omega f(y) \,\varphi_\#\mu(dy) = \int_\Omega f(\varphi(x))\, \mu(dx)
\qquad\text{for all }f\in C_b(\Omega).
\]


\bibliography{ref}
\bibliographystyle{alpha}

\printindex

\end{document}